\documentclass[pra, aps, twocolumn, preprintnumbers, superscriptaddress]{revtex4}
\usepackage{latexsym}
\usepackage{amsmath}
\usepackage{amssymb}
\usepackage{amsfonts}
\usepackage{ifthen}
\usepackage{revsymb}
\usepackage{yfonts}
\usepackage{graphicx}

\def\Real{\mathbb{R}}

\newcommand{\cancel}[1]{}
\newenvironment{sdp}[2]{
\smallskip
\begin{center}
\begin{tabular}{ll}
#1 & #2\\
subject to
}
{
\end{tabular}
\end{center}
\smallskip
}
\def\01{\{0,1\}}

\newcommand{\eps}{\varepsilon}
\newcommand{\ket}[1]{|#1\rangle}
\newcommand{\bra}[1]{\langle#1|}
\newcommand{\outp}[2]{|#1\rangle\langle#2|}

\newcommand\beq{\begin{equation}}
\newcommand\eeq{\end{equation}}
\newcommand\bea{\begin{eqnarray}}
\newcommand\eea{\end{eqnarray}}
\newcommand{\Tr}{\mbox{\rm Tr}}
\newcommand{\tr}{\mathrm{tr}}

\newcommand{\id}{\mathbb{I}}

\newtheorem{theorem}{Theorem}
\newtheorem{lemma}{Lemma}
\newtheorem{claim}{Claim}
\newtheorem{corollary}{Corollary}
\newtheorem{protocol}{Protocol}
\newtheorem{definition}{Definition}

\newenvironment{proof}
{\noindent {\bf Proof. }}
{{\hfill $\Box$}\\
 \smallskip}

\newcommand{\mM}{\mathcal{M}}
\newcommand{\mP}{\mathcal{P}}

\newcommand{\mX}{\mathcal{X}}
\newcommand{\mS}{\mathcal{S}}
\newcommand{\mN}{\mathcal{N}}

\newcommand{\setI}{\mathcal{I}}
\newcommand{\setT}{\mathcal{T}}
\newcommand{\setX}{\mathcal{X}}
\newcommand{\setF}{\mathcal{F}}

\newcommand{\set}[1]{\{#1\}}

\newcommand{\assign}{\ensuremath{\kern.5ex\raisebox{.1ex}{\mbox{\rm:}}\kern -.3em =}}

\renewcommand{\ol}[1]{\overline{#1}}

\begin{document}

\title{Cryptography from Noisy Storage}

\author{Stephanie \surname{Wehner}}
\email[]{wehner@caltech.edu}
\affiliation{Institute for Quantum Information, Caltech, 1200 E California Blvd, Pasadena, CA 91125, USA}
\author{Christian \surname{Schaffner}}
\email[]{c.schaffner@cwi.nl}
\affiliation{CWI, Kruislaan 413, 1098 SJ
Amsterdam, The Netherlands}
\author{Barbara M. \surname{Terhal}}
\affiliation{IBM, Watson Research Center, P.O. Box 218, Yorktown Heights, NY, USA}
\email[]{bterhal@gmail.com}
\date{\today}

\begin{abstract}
 We show how to implement cryptographic primitives based on the
  realistic assumption that quantum storage of qubits is noisy. We
  thereby consider individual-storage attacks, i.e.~the dishonest
  party attempts to store each incoming qubit separately. Our model is
  similar to the model of bounded-quantum storage, however, we
  consider an explicit noise model inspired by present-day technology.
  To illustrate the power of this new model, we show that a protocol
  for oblivious transfer (OT) is secure for \emph{any} amount of
  quantum-storage noise, as long as honest players can perform perfect
  quantum operations. Our model also allows the security of
  protocols that cope with noise in the operations of the honest
  players and achieve more advanced tasks such as secure
  identification.
\end{abstract}
\maketitle

Traditional cryptography is concerned with the secure and reliable
transmission of messages. With the advent of widespread electronic
communication new cryptographic tasks have become increasingly
important. Examples of such tasks are secure identification,
electronic voting, online auctions, contract signing and other
applications where the protocol participants do not necessarily
trust each other. It is well-known that almost all these interesting
tasks are impossible to realize without any restrictions on the
participating players, neither classically nor with the help of
quantum communication~\cite{lo:insecurity}. It is therefore an
important task to come up with a cryptographic model which restricts
the capabilities of adversarial players and in which these tasks
become feasible. It turns out that all such two-party protocols can
be based on a simple primitive called 1-2 Oblivious
Transfer~\cite{kilian:foundingOnOT} (1-2 OT), first introduced
in~\cite{wiesner:conjugate,rabin:ot,even:firstOT}. Hence, 1-2 OT is commonly used to provide a
``proof of concept'' for the universal power of a new model. In 1-2
OT, the sender Alice starts off with two bit strings $S_0$ and
$S_1$, and the receiver Bob holds a choice bit $C$.  The protocol
allows Bob to retrieve $S_C$ in such a way that Alice does not learn
any information about $C$ (thus, Bob cannot simply ask for $S_C$).
At the same time, Alice must be ensured that Bob only learns $S_C$,
and no information about the other string $S_{1-C}$ (thus, Alice
cannot simply send him both $S_0$ and $S_1$). A 1-2 OT protocol is
called unconditionally secure when neither Alice nor Bob can break
these conditions, even when given
unlimited resources.

In this letter, we propose a \emph{cryptographic model} based on
current practical and near-future technical limitations, namely that
quantum storage is noisy. Thus the presence of noise, the very
problem that makes it so hard to implement a quantum computer, can
actually be turned to our advantage. Recently it was shown that
secure OT is possible when the receiver Bob has a limited amount of
quantum memory~\cite{serge:bounded,serge:new} at his disposal.
Within this `bounded-quantum-storage model' OT can be implemented
securely as long as a dishonest receiver Bob can store at most
$n/4-O(1)$ qubits coherently, where $n$ is the number of qubits
transmitted from Alice to Bob. This approach assumes an explicit
limit on the physical number of qubits (or more precisely, on the rank
of the adversary's quantum state). However, at present we do not
know of any practical physical situation which enforces such a limit
for quantum information. We therefore propose an alternative model
of noisy quantum storage inspired by present-day physical
implementations: We require no explicit memory bound, but we assume
that any qubit that is placed into quantum storage undergoes a
certain amount of noise. The advantage of our model is that we can
evaluate the security parameters of a protocol explicitly in terms
of the noise. In this letter, we show that the OT protocol
from~\cite{serge:new} is secure in our new model. This simple OT
protocol could be implemented using photonic qubits (using
polarization or phase-encoding) with standard BB84 quantum key
distribution~\cite{BB:84, GRTZ:qkd_review} hardware, only with
different classical post-processing.

We analyze the case where the adversary performs individual-storage
attacks.  More precisely, Bob may choose to (partially) measure (a
subset of) his qubits immediately upon reception using an error-free
{\em product} measurement. In addition he can store each incoming
qubit, or post-measurement state from a prior partial measurement,
separately and wait until he gets additional information from Alice
(at Step~3 in Protocol~1). Once he obtained the additional
information he may perform an arbitrary coherent measurement on his
stored qubits using the stored classical data. We thereby assume that
qubit $q_i$ undergoes some noise while in storage, and we also
assume that the noise acts independently on each qubit.  In the
following, we use the super-operator $\mS_i$ to denote the combined
channel given by Bob's initial (partial) measurement and the noise.
Practically, noise can arise as a result of transferring the qubit
onto a different physical carrier, such as an atomic ensemble or
atomic state for example, or into an error-correcting code with
fidelity less than 1. In addition, the (encoded) qubit will undergo
noise once it has been transferred into `storage'. Hence, the
quantum operation $\mS_i$ in any real world setting necessarily
includes some form of noise.

First, we show that for any initial measurement, and \emph{any} noisy
superoperator $\mS_i$ the 1-2 OT protocol is secure if the honest
participants can perform {\em perfect} noise-free quantum
operations. As an explicit example we consider the case of
depolarizing noise during storage. In particular, we can show the
following all-or-nothing result: if Bob's storage noise is above a
certain threshold, his optimal cheating strategy is to perform a
measurement in the so-called Breidbart basis. On the other hand, if
the noise level is below the threshold, he is best off storing each
qubit as is.

Second, we consider a more practical setting using photonic qubits
where the honest participants experience noise themselves: their
quantum operations may be inaccurate or noisy, they may use weak
laser pulses instead of single photon sources, and qubits may
undergo decoherence during transmission. Note, however, that unlike
in QKD, we typically want to execute such protocols over very short
distances (for example in banking applications) where the
depolarization rate during transmission is very low.  We give a
practical OT-protocol that is a small modification of the perfect
protocol. It allows us to to deal with erasure errors (i.e.~photon
loss) separately.
We show how to derive trade-offs between the amount of storage
noise, the amount of noise for the operations performed by the
honest participants, and the security of the protocol.

Finally, we
briefly discuss the security of our protocol from the
future perspective of fault-tolerant quantum computation with photonic qubits.
We also discuss the issue of analyzing fully
coherent attacks for our protocol. Indeed, there is a close relation between the OT
protocol and BB84 quantum key distribution.
Our security analysis can in principle be carried over to obtain a
secure identification scheme in the noisy-quantum-storage model
analogous to~\cite{DFSS:secureid}. This scheme achieves
password-based identification and is of particular practical
relevance as it can be used for banking applications. \\

\subsection{Related work}
Precursors of the idea of basing
cryptographic security on storage-noise are already present
in~\cite{crepeau:practicalOT}, but no rigorous analysis was carried
through in that paper. Furthermore, it was pointed out
in~\cite{chris:diss,DFSS08journal} how the original
bounded-quantum-storage analysis applies in the case of noise levels
which are so large that the rank of a dishonest player's quantum
storage is reduced to $n/4$. In contrast, we are able to give an
explicit security trade-off even for small amounts of noise.  We
note that our security proof does not exploit the noise in the
communication channel (which has been done in the classical setting
to achieve cryptographic tasks, see e.g.~\cite{crepeau:weakenedOT,CMW04}), but is
solely based on the fact that the dishonest receiver's quantum
storage is noisy. A model based on classical noisy storage is akin
to the setting of a classical noisy channel, if the operations are
noisy, or the classical bounded-storage model, both of which are
difficult to enforce in practise. Another technical limitation has
been considered in \cite{salvail:physical} where a bit-commitment
scheme was shown secure under the assumption that the dishonest
committer can only measure a limited amount of qubits coherently.
Our analysis differs in that we can in fact allow any coherent
destructive measurement at the end of the protocol.

\section{Definitions and Tools}
We start by introducing some tools, definitions and technical lemmas.
To define the security of OT we need to express what it means for a
dishonest quantum player not to gain any information.  Let $\rho_{XE}$ be a
state that is part classical, part quantum, i.e.~a cq-state
$\rho_{XE}=\sum_{x \in \setX} P_X(x) \outp{x}{x} \otimes \rho_E^x$.
Here, $X$ is a classical random variable distributed over the finite
set $\setX$ according to distribution $P_{X}$. The {\em non-uniformity} of $X$ given
$\rho_E = \sum_x P_X(x) \rho_E^x$
is defined as
\begin{equation}
d(X|\rho_E) := \frac{1}{2}||\,\id/|\mX| \otimes \rho_E-\sum_{x}P_{X}(x) \outp{x}{x} \otimes \rho_{E}^x\,||_{\tr},
\end{equation}
where
$||A||_{\tr} = \Tr\sqrt{A^\dagger A}$. Intuitively, if
$d(X|\rho_E) \leq \eps$ the distribution of $X$ is $\eps$-close to
uniform even given $\rho_E$, i.e., $\rho_E$ gives hardly any
information about $X$.
A simple property of the non-uniformity which follows from its definition is that
\begin{equation}
d(X|\rho_{ED})=d(X|\rho_E)
\label{eq:indep}
\end{equation}
for any cq-state of the form $\rho_{XED}=\rho_{XE} \otimes \rho_D$.

We prove the security of a randomized version of OT. In such a
protocol, Alice does not choose her input strings herself, but
instead receives two strings $S_0$, $S_1 \in \01^\ell$ chosen
uniformly at random by the protocol. Randomized OT (ROT) can easily
be converted into OT: after the ROT protocol is completed, Alice uses
her strings $S_0,S_1$ obtained from ROT as one-time pads to encrypt
her original inputs $\hat{S_0}$ and $\hat{S_1}$, i.e.~she sends an
additional classical message consisting of $\hat{S_0} \oplus S_0$
and $\hat{S_1} \oplus S_1$ to Bob. Bob can retrieve the message of
his choice by computing $S_C \oplus (\hat{S}_C \oplus S_C) =
\hat{S}_C$. He stays completely ignorant about the other message
$\hat{S}_{1-C}$ since he is ignorant about $S_{1-C}$. The security
of a quantum protocol implementing ROT is formally defined in \cite{serge:bounded,serge:new}:
\begin{definition} \label{def:ROT}
An $\eps$-secure 1-2 $\mbox{ROT}^\ell$ is a protocol between Alice
and Bob, where Bob has input $C \in \01$, and Alice has no input.
For any distribution of $C$:
\begin{itemize}
\item (Correctness) If both parties are honest, Alice gets output $S_0,S_1 \in \01^\ell$ and
Bob learns $Y = S_C$ except with probability $\eps$.
\item (Receiver-security) If Bob is honest and obtains output $Y$, then for any cheating strategy of Alice resulting in her state $\rho_A$,
there exist random variables $S'_0$ and $S'_1$ such that $\Pr[Y=S'_C] \geq 1- \eps$ and
$C$ is independent of $S'_0$,$S'_1$ and $\rho_A$.%
\item (Sender-security) If Alice is honest, then for any cheating strategy of Bob resulting in his state $\rho_B$,
there exists a random variable $C' \in \01$ such that $d(S_{1-C'}|S_{C'}C'\rho_B) \leq \eps$.
\end{itemize}
\end{definition}

The OT protocol makes use of two-universal hash functions.
These hash functions are used for privacy amplification similar as in
quantum key distribution. A class $\setF$ of functions $f: \01^n
\rightarrow \01^\ell$ is called two-universal if for all $x\neq y \in
\01^n$ and $f \in \setF$ chosen uniformly at random from $\setF$, we have
$\Pr[f(x) = f(y)] \leq 2^{-\ell}$.  For example, the set of all affine
functions from $\01^n$ to $\01^\ell$ is two-universal~\cite{CarWeg79}.  The following
theorem expresses how hash functions can increase the privacy of a
random variable X given a quantum adversary holding $\rho_E$ and the
function $F$:

\begin{theorem}[Th. 5.5.1 in~\cite{renato:diss} (see also~\cite{renato:compose})]
Let $\setF$ be a class of two-universal hash functions from $\01^n$ to $\01^\ell$.
Let $F$ be a random variable that is uniformly and independently distributed over $\setF$,
and let $\rho_{XE}$ be a cq-state. Then,
$$
d(F(X)|F,\rho_E) \leq 2^{-\frac{1}{2}\left(H_2(X|\rho_E) - \ell\right)-2},
$$
where $H_2(\cdot|\cdot)$ denotes the conditional collision entropy defined in~\cite{renato:diss}
as $H_2(X|\rho_E) := - \log \Tr((\id \otimes \rho_E^{-\frac{1}{2}}) \rho_{XE})^2$
of the cq-state $\rho_{XE}$.
\end{theorem}

In our application we will make use of a simplified form of this theorem which follows directly
from~\cite[Lemma 1]{wehner06d}. The non-uniformity in the theorem above is bounded by the average success
probability of guessing $x$ given the state $\rho_E$:
\begin{lemma}\label{simplifiedPA}
For a measurement $M$ with POVM elements $\{M_x\}_{x \in {\cal X}}$
let $p_{y|x}^M={\rm Tr} M_y \rho_E^x$ the probability of outputting
guess $y$ given $\rho_E^x$. Then $P_g(X|\rho_E) = \sup_M \sum_x
P_X(x) p_{x|x}^M$ is the maximal average success probability of
guessing $x \in \setX$ given the reduced state $\rho_E$ of the
cq-state $\rho_{XE}$. We have
$$
d(F(X)|F,\rho_E) \leq 2^{\frac{\ell}{2}-1} \sqrt{P_g(X|\rho_E)} \, .
$$
If we have an additional $k$ bits of classical information $D$ about $X$,
we can bound
\begin{equation}
d(F(X)|F,D,\rho_E) \leq 2^{\frac{\ell+k}{2}-1}\sqrt{P_g(X|\rho_E)} \, .
\label{eq:addinfo}
\end{equation}
\end{lemma}

The following lemma is proven in the Appendix and states that the
optimal strategy to guess $X=x \in \set{0,1}^n$ given individual
quantum information about the bits of $X$ is to measure each
register individually.
\begin{lemma}\label{lemma:individual}
Let $\rho_{XE}$ be a cq-state with uniformly distributed $X=x \in
\01^n$ and $\rho_E^x = \rho_{E_1}^{x_1} \otimes \ldots \otimes
\rho_{E_n}^{x_n}$. Then the maximum probability of guessing $x$
given state $\rho_E$ is $P_g(X|\rho_E)  = \Pi_{i=1}^n
P_g(X_i|\rho_{E_i})$, which can be achieved by measuring each
register separately.
\end{lemma}

The last tool we need is an uncertainty relation for noisy channels
and measurements. Let $\sigma_{0,+} = \outp{0}{0}$, $\sigma_{1,+} =
\outp{1}{1}$, $\sigma_{0,\times} = \outp{+}{+}$ and $\sigma_{1,\times} =
\outp{-}{-}$ denote the BB84-states corresponding to the encoding of a bit $z
\in \01$ into basis $b \in \{+,\times\}$ (computational resp.
Hadamard basis). Let $\sigma_+ = (\sigma_{0,+} + \sigma_{1,+})/2$
and $\sigma_\times = (\sigma_{0,\times} + \sigma_{1,\times})/2$.
Consider the state ${\cal S}(\sigma_{z,b})$ for
some super-operator ${\cal S}$.
Note that $P_g(X|{\cal S}(\sigma_{b}))$ (see Lemma~\ref{lemma:individual})
denotes the maximal average success probability for guessing a uniformly distributed $X$ when $b=+$ or
$b=\times$. An uncertainty relation for such success probabilities can
be stated as
\begin{equation}
P_g(X|{\cal S}(\sigma_+)) \cdot P_g(X|{\cal S}(\sigma_\times)) \leq \Delta({\cal S})^2,
\label{eq:uncert}
\end{equation}
where $\Delta$ is a function from the set of superoperators to the real numbers.
For example, when ${\cal S}$ is a quantum measurement ${\cal M}$
mapping the state $\sigma_{z,b}$ onto purely classical information
it can be argued (e.g.~by using a purification argument and
Corollary 4.15 in \cite{chris:diss}) that $\Delta({\cal M}) \equiv
\frac{1}{2}(1+2^{-1/2})$ which can be achieved by a measurement in
the Breidbart basis,
where the Breidbart basis is given by $\{\ket{0}_B,\ket{1}_B\}$ with $\ket{0}_B =
\cos(\pi/8)\ket{0} + \sin(\pi/8) \ket{1}$ and $\ket{1}_B =
\sin(\pi/8) \ket{0} - \cos(\pi/8) \ket{1}$.

It is clear that for a unitary superoperator $U$ we have
$\Delta(U)^2=1$ which can be achieved. It is not hard to show that
(see the proof in the Appendix)
\begin{lemma}\label{lem:iso}
The only superoperators ${\cal S}\colon \mathbb{C}_2 \rightarrow \mathbb{C}_k$ for which
\begin{equation}
P_g(X|{\cal S}(\sigma_+)) \cdot P_g(X|{\cal S}(\sigma_{\times}))=1,
\end{equation}
are reversible operations.
\end{lemma}

\section{Protocol and Analysis}
We use $\in_R$ to denote the uniform choice of an element from a set.
We further use $x_{|\setT}$ to denote the string $x=x_1,\ldots,x_n$
restricted to the bits indexed by the set $\setT \subseteq
\{1,\ldots,n\}$.  For convenience, we take $\{+,\times\}$ instead of
$\01$ as domain of Bob's choice bit $C$ and denote by $\overline{C}$ the
bit different from $C$.
\begin{protocol}[\cite{serge:new}]1-2 $\mbox{ROT}^\ell(C,T)$
\begin{enumerate}
\item Alice picks $X \in_R \01^n$ and $\Theta \in_R \{+,\times\}^n$.
  Let $\setI_b = \{i\mid \Theta_i = b\}$ for $b \in \{+,\times\}$. At
  time $t=0$, she sends $\sigma_{X_1,\Theta_1} \otimes \ldots \otimes
  \sigma_{X_n,\Theta_n}$ to Bob.
\item Bob measures all qubits in the basis corresponding to his choice
  bit $C \in \{+,\times\}$.  This yields outcome $X' \in \01^n$.
\item Alice picks two hash functions $F_+,F_\times \in_R
  \setF$, where $\setF$ is a class of two-universal hash functions. At time $t=T$, she sends $\setI_+$,$\setI_\times$, $F_+$,$F_\times$ to Bob.
Alice outputs $S_+ = F_+(X_{|\setI_+})$ and $S_{\times}
  = F_\times(X_{|\setI_\times})$ \footnote{If $X_{|\setI_b}$ is less
    than $n$ bits long Alice pads the string $X_{|\setI_b}$ with 0's
    to get an $n$ bit-string in order to apply the hash function to
    $n$ bits.}.
\item Bob outputs $S_C = F_C(X'_{|\setI_C})$.
\end{enumerate}
\end{protocol}

\subsection{Analysis}
We first show that this protocol is secure according to
Definition~\ref{def:ROT}.

(i) correctness: It is clear that the protocol is correct.  Bob can
determine the string $X_{|\setI_C}$ (except with negligible
probability $2^{-n}$ the set ${\cal I}_C$ is non-empty) and hence
obtains $S_C$.

(ii) security against dishonest Alice: this holds in the same way as
shown in~\cite{serge:new}. As the protocol is non-interactive, Alice
never receives any information from Bob at all, and Alice's input
strings can be extracted by letting her interact with an unbounded
receiver.

(iii) security against dishonest Bob: Our goal is to show that there
exists a $C' \in \{+,\times\}$ such that Bob is completely ignorant
about $S_{\ol{C'}}$. In our model Bob's collective storage cheating
strategy can be described by some super-operator
$\mS=\bigotimes_{i=1}^n \mS_i$ that is applied on the qubits between
the time they arrive at Bob's and the time $T$ that Alice sends the
classical information.
We define the choice bit $C'$ as a fixed function of $\mS$. 
Formally, we set $C' \equiv +$ if
$\prod_{i=1}^n P_g(X_i|\mS_i(\sigma_+)) \geq \prod_{i=1}^n
P_g(X_i|\mS_i(\sigma_{\times}))$ and $C' \equiv \times$ otherwise.

Due to the uncertainty relation for each ${\cal S}_i$ (from
Eq.~(\ref{eq:uncert})) it then holds that $\prod_i P_g(X_i|{\cal
S}_i(\sigma_{\ol{C'}})) \leq \prod_i \Delta({\cal S}_i) \leq
(\Delta_{\rm max})^n$ where $\Delta_{\rm max} \assign \max_{i} \Delta({\cal
S}_i)$. This will be used in the proof below.

In the remainder of this section, we show that the non-uniformity
$\delta_{\rm sec} \assign d(S_{\ol{C'}}|S_{C'}C'\rho_B)$ is negligible
in $n$ for collective attacks. Here $\rho_B$ is the complete quantum
state of Bob's lab at the end of the protocol including the classical
information $\setI_+, \setI_{\times}, F_+, F_{\times}$ he got from
Alice and his quantum information $\bigotimes_{i=1}^n
\mS_i(\sigma_{X_i,\Theta_i})$. Expressing the non-uniformity in terms
of the trace-distance allows us to observe that $\delta_{\rm
  sec}=2^{-n} \sum_{\theta \in \{+,\times\}^n} d(S_{\ol{C'}} |
\Theta=\theta , S_{C'} C' \rho_B)$.  Now, for fixed $\Theta=\theta$,
it is clear from the construction that $S_{C'},C',F_{C'}$ and
$\bigotimes_{i \in \setI_{C'}} \mS_i(\sigma_{X_i,C'})$ are independent
of $S_{\ol{C'}}=F_{\ol{C'}}(X_{|\setI_{\ol{C'}}})$ and we can use
Eq.~(\ref{eq:indep}).  Hence, one can bound the non-uniformity as in
Lemma~\ref{simplifiedPA}, i.e.~by the square-root of the probability
of correctly guessing $X_{|_{\setI_{\ol{C'}}}}$ given the state
$\bigotimes_{i \in \setI_{\ol{C'}}} \mS_i (\sigma_{X_i,\ol{C'}})$.
Lemma~\ref{lemma:individual} tells us that to guess $X$, Bob can
measure each remaining qubit individually and hence we obtain
\begin{align*}
\delta_{\rm sec} &\leq 2^{\frac{\ell}{2}-1} \cdot 2^{-n} \!\!\! \sum_{\theta \in \{+,\times\}^n}
\sqrt{ \prod_{i \in \setI_{\ol{C'}}} P_g(X_i | \mS_i(\sigma_{\ol{C'}}) )}\\
&\leq 2^{\frac{\ell}{2}-1} \sqrt{ 2^{-n} \sum_{\theta \in \{+,\times\}^n}
\prod_{i \in \setI_{\ol{C'}}} P_g(X_i | \mS_i(\sigma_{\ol{C'}}) )}\\
&= 2^{\frac{\ell}{2}-1} \sqrt{ 2^{-n} \prod_{i=1}^n \big(1+ P_g(X_i |
  \mS_i(\sigma_{\ol{C'}}) ) \big)} \, ,
\end{align*}
where we used the concavity of the square-root function in the last
inequality.  Lemma~\ref{lem:bestdistribution} together with the bound
$\prod_i P_g(X_i|{\cal S}_i(\sigma_{\ol{C'}})) \leq (\Delta_{\rm max})^n$
lets us conclude that
\[ \delta_{\rm sec} \leq 2^{\frac{\ell}{2}-1} \cdot (\Delta_{\rm
  max})^{\frac{\log(4/3)}{2} n} \, .
\]

Lemma~\ref{lem:iso} shows that for essentially any noisy superoperator $\Delta({\cal S}) < 1$.
This shows that for any collective attacks there exists an $n$ which
yields arbitrarily high security.

\subsection{Example}
Let us now consider the security in an explicit example: a noisy depolarizing channel.
In order to explicitly bound $\Delta(\mS_i)$ we should allow for intermediate strategies of Bob in
which he partially measures the incoming qubits leaving some quantum information undergoing depolarizing noise.
To model this noise we let $\mS_i =\mN
\circ \mP_i$, where $\mP_i$ is any noiseless quantum operation of Bob's
choosing from one qubit to one qubit that generates some classical
output. For example, $\mP_i$ could be a partial measurement providing
Bob with some classical information and a slightly disturbed quantum
state, or just a unitary operation. Let
$$
\mN(\rho) := r \rho + (1-r)\frac{\id}{2}
$$
be the fixed depolarizing 'quantum storage' channel that Bob cannot
influence. (see Figure \ref{figure:depolModel})

\begin{figure}
\includegraphics{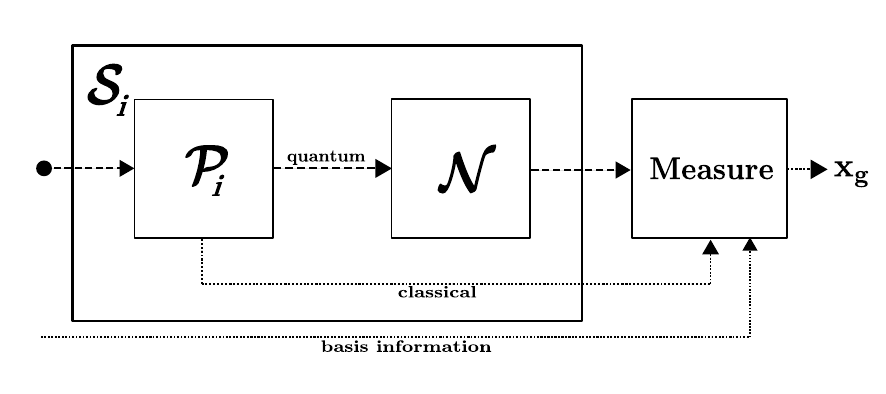}
\caption{Bob performs a partial measurement $\mP_i$, followed by noise $\mN$, and outputs a guess bit
$x_g$ depending on his classical measurement outcome, the remaining quantum state, and the
additional basis information.} \label{figure:depolModel}
\end{figure}

To determine $\delta_{\rm sec}$, we have to find an uncertainty relation
similar to Eq.~(\ref{eq:uncert}) by optimizing over all possible
partial measurements $\mP_i$
$$
\max_{\mS_i} \Delta(\mS_i)^2 = \max_{\mP_i} P_g(X|{\cal S}_i(\sigma_+)) \cdot
P_g(X|{\cal S}_i(\sigma_\times)).
$$

We solve this problem for depolarizing noise using
the symmetries inherent in our problem. In Appendix \ref{app:depol} we prove the following.
\begin{theorem} \label{thm:depolarize}
Let $\mN$ be the depolarizing channel and let $\max_{\mS_i} \Delta(\mS_i)$
be defined as above. Then
$$
\max_{\mS_i} \Delta(\mS_i) = \left\{
\begin{array}{ll}
\frac{1+r}{2} & \mbox{ for } r \geq \frac{1}{\sqrt{2}}\\
\frac{1}{2} + \frac{1}{2\sqrt{2}} & \mbox{ for } r <
\frac{1}{\sqrt{2}}
\end{array}\right.
$$
\end{theorem}
Our result shows that for $r < 1/\sqrt{2}$ a direct measurement $\mM$ in the Breidbart basis is the best attack Bob can perform. For this measurement,
we have $\Delta(\mM)= 1/2 +
1/(2\sqrt{2})$.
If the depolarizing noise is low ($r \geq 1/\sqrt{2}$), then our result
states that the best strategy for Bob is to simply store the qubit as is.

\section{Practical Oblivious Transfer}

In this section, we prove the security of a ROT protocol that is
robust against noise for the honest parties. Our protocol is thereby
a small modification of the protocol considered
in~\cite{chris:diss}. Note that for our analysis, we have to assume
a worst-case scenario where a dishonest receiver Bob has access to a
perfect noise-free quantum channel and only experiences noise during
storage. First, we consider erasure noise (in practice corresponding
to photon loss) during preparation, transmission and measurement of
the qubits by the honest parties. Let $1-p_{\rm erase}$ be the total
probability for an honest Bob to measure and detect a photon in the
$\{+,\times\}$ basis given that an honest Alice prepares a weak
pulse in her lab and sends it to him. The probability $p_{\rm
erase}$ is determined among others by the mean photon number in the
pulse, the loss on the channel and the quantum efficiency of the
detector. In our protocol we assume that the (honest) erasure rate
$p_{\rm erase}$ is {\em independent} of whether qubits were encoded
or measured in the $+$- or $\times$-basis. This assumption is
necessary to guarantee the correctness and the security against a
cheating \emph{Alice} only. Fortunately, this assumption is well matched with
physical capabilities.

Any other noise source during preparation, transmission and
measurement can be characterized as an effective classical noisy
channel resulting in the output bits $X'$ that Bob obtains at
Step~\ref{step:reception} of Protocol~\ref{prot:practical}.  For simplicity, we model
this compound noise source as a classical binary symmetric channel
acting independently on each bit of $X$. Typical noise sources for
polarization-encoded qubits are depolarization during transmission,
dark counts in Bob's detector and misaligned polarizing
beam-splitters. Let the effective bit-error probability of this binary
symmetric channel be $p_{\rm error} < 1/2$.

Before engaging in the actual protocol, Alice and Bob agree on the
system parameters $p_{\rm erase}$ and $p_{\rm error}$ similarly to
Step~1 of the protocol in~\cite{crepeau:practicalOT}. Furthermore,
they agree on a family $\set{C_n}$ of linear error correcting codes
of length $n$ capable of efficiently correcting $n \cdot p_{\rm
error}$ errors. For any string $x \in \set{0,1}^n$, error correction
is done by sending the syndrome information $syn(x)$ to Bob from
which he can correctly recover $x$ if he holds an output $x' \in
\set{0,1}^n$ obtained by flipping each bit of $x$ independently with
probability $p_{\rm error}$. It is known that for large enough $n$,
the code $C_n$ can be chosen such that its rate is arbitrarily close
to $1-h(p_{\rm error})$ and the syndrome length (the number of parity check bits) are
asymptotically bounded by $|syn(x)| < h(p_{\rm error})
n$~\cite{crepeau:efficientOT}, where $h(p_{\rm error})$ is the binary Shannon
entropy. We assume the players have synchronized clocks. In each
time slot, Alice sends one qubit (laser pulse) to Bob.

\begin{protocol}Noise-Protected Photonic 1-2 $\mbox{ROT}^\ell(C,T)$ \label{prot:practical}
\begin{enumerate}
\item Alice picks $X \in_R \01^n$ and $\Theta \in_R \{+,\times\}^n$.
\item For $i=1,\ldots,n$: In time slot $t=i$, Alice sends
  $\sigma_{X_i,\Theta_i}$ as a phase- or polarization-encoded weak pulse
  of light to Bob.
\item \label{step:reception} In each time slot, Bob measures the
  incoming qubit in the basis corresponding to his choice bit $C \in
  \{+,\times\}$ and records whether he detects a photon or not.  He
  obtains some bit-string $X' \in \01^m$ with $m \leq n$.
\item Bob reports back to Alice in which time slots he received a
  qubit. Alice restricts herself to the set of $m < n$ bits that Bob
  did not report as missing. Let this set of qubits be $S_{\rm
    remain}$ with $|S_{\rm remain}|=m$.
\item \label{step:alicecheck} Let $\setI_b = \{i \in S_{\rm
    remain}\mid \Theta_i = b\}$ for $b \in \{+,\times\}$ and let
  $m_b=|\setI_b|$.  Alice aborts the protocol if either $m_+$ or
  $m_\times \leq (1-p_{\rm erase})n/2-O(\sqrt{n})$.
  If this
  is not the case, Alice picks two two-universal hash functions
  $F_+,F_\times \in_R \setF$. At time $t=n+T$, Alice sends
  $\setI_+$,$\setI_\times$, $F_+$,$F_\times$, and the syndromes $syn(X_{|\setI_+})$
  and $syn(X_{|\setI_\times})$ according to codes of appropriate length $m_b$ to Bob.  Alice outputs $S_+ =
  F_+(X_{|\setI_+})$ and $S_{\times} = F_\times(X_{|\setI_\times})$.
\item Bob uses $syn(X_{|\setI_C})$ to correct the errors on his output $X'_{|\setI_C}$. He obtains the corrected bit-string
$X_{\rm cor}$ and outputs $S'_C = F_C(X_{\rm cor})$.
\end{enumerate}
\end{protocol}

Let us consider the security and correctness of this modified protocol.\\
(i) correctness: By assumption, $p_{\rm erase}$ is
independent of the basis in which Alice sent the qubits.
Thus, $S_{\rm remain}$ is with high
probability a random subset of $m \approx (1-p_{\rm erase})n \pm
O(\sqrt{n})$ qubits independent of the value of $\Theta$. This
implies that in Step~\ref{step:alicecheck} the protocol is aborted with a probability
exponentially small in $m$, and hence in $n$. The codes are chosen
such that Bob can decode except with negligible probability. These
facts imply that if both parties are honest the protocol is correct
(i.e. $S_C=S'_C$) with exponentially small
probability of error. \\
(ii) security against dishonest Alice: Even though in this scenario Bob {\em does}
communicate to Alice, the information stating which qubits were erased is (by assumption) independent of the basis in which he
measured and thus of his choice bit $C$. Hence Alice does not learn
anything about his choice bit $C$. Her input strings can be extracted as in Protocol~1. \\
(iii) security against dishonest Bob: First of all, we note that Bob
can always make Alice abort the protocol by reporting back an
insufficient number of received qubits. If this is not the case,
then we define $C'$ as in the analysis of Protocol~1 and we need to
bound the non-uniformity $\delta_{\rm sec}$ as before. Let us for
simplicity assume that $m_b=m/2$ (this is true with high
probability, modulo $O(\sqrt{n})$ factors which become negligible in
the security for large $n$) with $m \approx (1-p_{\rm erase})n$
We now follow through the same analysis, where we restrict ourselves
to the set of remaining qubits.
We first follow through the same steps simplifying the non-uniformity
using that the total attack superoperator ${\cal S}$ is a product of
superoperators. Then we use the bound in Lemma~\ref{simplifiedPA}
for each $\theta \in \{+,\times\}^n$ where we now have to condition on the additional information
$syn(X_{|\setI_{\ol{C'}}})$ which is $m h(p_{\rm error})/2$ bits long.
Using Eq. (\ref{eq:addinfo}) and following identical steps in the
remainder of the proof implies
\begin{equation}
\delta_{\rm sec} \leq  2^{\frac{\ell}{2}-1 + h(p_{\rm error}) \frac{m}{4}}
(\Delta_{\rm
  max})^{\frac{\log(4/3)}{2} m} \, .
  \label{eq:sec}
\end{equation}
From this expression it is clear that the security depends crucially
on the value of $\Delta_{\rm max}$ versus the binary entropy $h(p_{\rm error})$.
The trade-off in our bound is not extremely favorable for security
as we will see.

\subsection{Depolarizing noise}

We first consider again the security tradeoff when Bob's storage is
affected by depolarizing noise, and additionally the channel itself is
subject to depolarizing noise. Let us assume that $r <
1/\sqrt{2}$ for the storage noise.  According to
Theorem~\ref{thm:depolarize}, Bob's \emph{optimal} attack is to
measure each qubit individually in the Breidbart basis.  In this case,
our protocol is secure as long as $h(p_{\rm error}) < 2 \log(\frac12 +
\frac{1}{2 \sqrt{2}}) \log(3/4)$.  Hence, we require that $p_{\rm
  error} \lessapprox 0.029$.  This puts a strong restriction on the
noise rate of the honest protocol. Yet, since our protocols are
particularly interesting at short distances (e.g. in the case of
secure identification), we can imagine very short free-space
implementations such that depolarization noise during transmission is
negligible and the main depolarization noise source is due to Bob's
honest measurements.

In the near-future we may anticipate that storage is better than direct measurement when good photonic memories
become available (\cite{julsgaard+:mem, boozer+:qmem, chaneliere+:qmem, eisaman+:qmem, rosenfeld+:qmem,
  pittman_franson:memory}).
However, we are free in our protocol to stretch the waiting time $T$ between Bob's reception of the qubits and his reception of the classical basis information, say, to seconds, which means
that one has to consider the overall noise rate on a qubit that is stored for seconds. Clearly, there is 
a strict tradeoff between the noise $p_{\rm error}$ on the channel experienced by the honest
parties, and the noise experienced by dishonest Bob.

For $r \geq 1/\sqrt{2}$ (when storage is better than the Breidbart attack) we also obtain a tradeoff
involving $r$. Suppose that the qubits in the honest protocol are also subjected to depolarizing
noise at rate $1-r_{\rm honest}$. The effective classical error rate for a depolarizing channel
is then simply $p_{\rm error}=(1-r_{\rm honest})/2$. Thus we can consider when the function
$h(p_{\rm error})/4+\log(\frac{1+r}{2})\log(4/3)/2$ goes below 0. If we assume that
$r_{\rm honest} = a r \leq 1$, for some scaling factor $1 \leq a \leq 1/r$
(i.e., the honest party never has more noise than the dishonest party),
we obtain a clear tradeoff between $a$ and $r$ depicted in Figure~\ref{figure:tradeoff}.

\begin{figure}
\includegraphics[scale=0.5]{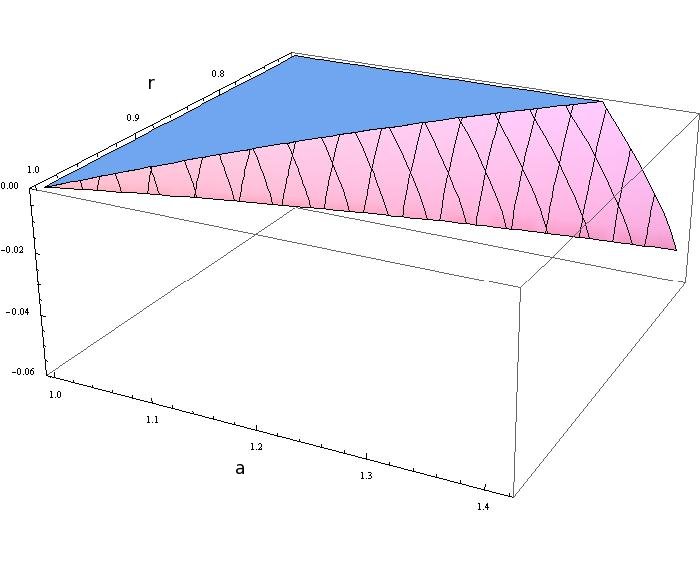}
\caption{
$h((1 - a r)/2)/4+\log(\frac{1+r}{2})\log(4/3)/2$, where we only show the region below 0, i.e.,
where security can be attained.}
\label{figure:tradeoff}
\end{figure}

\subsection{Other Attacks}
In a practical setting, other attacks may be possible which are not
captured by the model we used when analyzing depolarizing noise. For
example, attacks that relate to the protocol being implemented with
weak coherent states.
We discuss the affect of such practical problems in this section, but do not
claim to prove security of the practical protocol in full
generality. Instead, we merely discuss several practical attacks that
a dishonest Bob may mount.

Let us consider the security threat that comes from using coherent
weak laser pulses. For a mean photon number $\mu$, the probability to
have more than one photon in the beam is $P(k > 1) \approx \mu/2$
\cite{GRTZ:qkd_review}, where $k$ is the number of photons and $P(k)$
is the probability of $k$ photons in the beam with mean photon number $\mu$. 
In principle, this implies that Bob can measure
in both bases with probability $\mu/2$ (and he knows when this
occurs). If with remaining probability $1-\mu/2$ he is able to do a
measurement in the Breidbart basis, then for such attack we have
$\Delta_{\rm bm}=\mu/2+(1-\mu/2)(1/2+1/(2
\sqrt{2}))=1/2+1/(2\sqrt{2})+\mu(1-1/\sqrt{2})/4$.

Another attack is the following. Upon reception of his qubits Bob
tries to beam-split each incoming pulse and measure the outgoing
modes in both bases. In case he does not succeed he would like to
declare erasures. In Step~5 of the protocol Alice aborts the
protocol when Bob declares too many erasures: in principle, this 
can prevent Bob from making the protocol completely unsafe with this attack. Such a
beam-splitting attack does however put another constraint on the
region of error rates where one can have security using Eq.~\eqref{eq:sec}. Let us sketch the security bound for this particular attack.
Among the $m=(1-p_{\rm erase}) n$ remaining time slots, Bob will have $P(k > 1)p_{\rm
beamsplit} n \approx n\mu/4$ slots where he gets two or more
photons and measures them successfully in both bases (assuming perfect detector
efficiency), where $p_{\rm beamsplit} = 1/2$. 
For these slots, $\Delta=1$ so they do not enter the
security bound. For the $n(1-p_{\rm erase}-\mu/4)$ remaining time slots, he
is in a situation similar to before. Let us assume that the erasure
rate $p_{\rm erase}\approx P(k=0)+P(k \geq 1)p_{\rm no detect}$ where
$p_{\rm no detect}$ is the probability that Bob does not detect a
photon with his devices. Since the probability of
emitting a very large number of photons is small, we approximate the true
value by letting $p_{\rm no detect}$ be independent of $k$.
We have
$P(0)=e^{-\mu} \approx 1-\mu$ for small $\mu$ and thus $n((1-p_{\rm
erase})-\mu/4)=n \mu (p_{\rm detect}-1/4)$. In principle, this leads
to a bound as in Eq.~(\ref{eq:sec}).
However, security remains to be analyzed rigorously, and one needs
to determine Bob's optimal cheating strategy.
If single photon sources were used, such attacks could be
excluded.

In our analysis, we assumed that Alice and Bob can reliably establish a
bound on $p_{\rm erase}$. However $p_{\rm erase}$ may contain a
sizable contribution from the quantum efficiency of the detectors used
by Bob and a dishonest receiver may cheat by using better detectors
than he tells Alice during the error estimation process. 
For example, in the extreme case he could
convince Alice that his devices are so bad that of the $n$ inputs he
can detect a photon only in $\mu n/4$ cases. If instead he has perfect
devices and measures two photons successfully in both bases $\mu n/4$
times, he made the protocol completely insecure. Thus we assume in our
protocol that Alice can establish a reliable and reasonable lower
bound on $p_{\rm erase}$.


For current and near-future implementations we note that an
important practical limitation on Bob's attacks is the following.
Since a photon measurement is destructive with current technology,
Bob cannot store his qubits while at the same time reporting
correctly which ones were erased. So if Bob wants to store his
qubits, he has to guess which qubits were erased. This implies
that among the set of qubits in the set ${\cal I}_b$ approximately
$p_{\rm erase}m_b$ are in fact erased. For an erasure channel
with rate $p_{\rm erase}$ it is simple to show that $\Delta(\mS_{\rm
erase})=1-p_{\rm erase}/2$. Since erasure rates can easily be high
(due to small $\mu$ and other sources of photon loss), say of
$O(10^{-1})$, this limits the threat of a storage attack within the
current technology setting.

\subsection{Fault-tolerant computation}
Let us discuss the long-term security when {\em fault-tolerant}
photonic computation would become available (with the KLM scheme
\cite{KLM:lo} for example). In such a scenario dishonest Bob can
encode the incoming quantum information into a fault-tolerant quantum
memory. This implies that in storage, the effective noise rate can be
made arbitrarily small. However, the encoding of a single unknown state
is not a fault-tolerant quantum operation: already the encoding process introduces
errors whose rates cannot be made arbitrarily small with increasing
effort. Hence, even in the presence of a quantum computer, there is
a residual storage noise rate due to the unprotected encoding
operation. The question of security then becomes a question of a
trade-off between this residual noise rate versus the intrinsic noise
rate. Our current security bound is too weak though, to show security
in such scenario.

\section{Conclusion}
We have determined security bounds for a perfect and a practical ROT
protocol given collective storage attacks by Bob.  Ideally, we would
like to be able to show security against general coherent noisy
attacks.  The problem with analyzing a coherent attack of Bob
described by some super-operator ${\cal S}$ affecting all his incoming
qubits is not merely a technical one: one first needs to determine a
realistic noise model in this setting.  It may be possible using de
Finetti theorems as in the proof of QKD \cite{renato:diss} to prove
for a symmetrized version of our protocol that any coherent attack by
Bob is equivalent to a collective attack. One can in fact analyze a
specific type of coherent noise, one that essentially corresponds to
an eavesdropping attack in QKD. Note that the 1-2 OT protocol can be
seen as two runs of QKD interleaved with each other. The strings
$f(x_{|\setI_+})$ and $f(x_{|\setI_\times})$ are then the two keys
generated. The noise must be such that it leaves Bob with exactly the
same information as the eavesdropper Eve in QKD. In this case, it
follows from the security of QKD that the dishonest Bob (learning
exactly the same information as the eavesdropper Eve) does not learn
anything about the two keys.

It is an important open question whether it is possible to derive
security bounds (or find a better OT protocol) which give better
trade-offs between noise in the honest protocol and noise induced by
dishonest Bob. Finally, it remains to address composability of the
protocol within our model, which has already been considered for the
bounded-quantum-storage model~\cite{WW07:compose}.

\acknowledgments We thank Charles Bennett, David DiVincenzo, Renato
Renner and Falk Unger for interesting discussions and Ronald de Wolf
for suggestions regarding Lemma~\ref{lem:bestdistribution}. We are
especially grateful 
to Hoi-Kwong Lo for bringing up attacks that relate to
the use of weak laser pulses in the practical OT protocol.
This work was completed while SW was a PhD student at CWI, Amsterdam, 
Netherlands.
CS and SW were supported by EU fifth framework project QAP IST 015848
and the NWO VICI project 2004-2009. BMT acknowledges support by DTO
through ARO contract number W911NF-04-C-0098. SW thanks IBM Watson
and BMT thanks the Instituut Lorentz in Leiden for their kind
hospitality. At both locations part of this work were completed.
\bibliographystyle{apsrev}

\begin{thebibliography}{39}
\expandafter\ifx\csname natexlab\endcsname\relax\def\natexlab#1{#1}\fi
\expandafter\ifx\csname bibnamefont\endcsname\relax
  \def\bibnamefont#1{#1}\fi
\expandafter\ifx\csname bibfnamefont\endcsname\relax
  \def\bibfnamefont#1{#1}\fi
\expandafter\ifx\csname citenamefont\endcsname\relax
  \def\citenamefont#1{#1}\fi
\expandafter\ifx\csname url\endcsname\relax
  \def\url#1{\texttt{#1}}\fi
\expandafter\ifx\csname urlprefix\endcsname\relax\def\urlprefix{URL }\fi
\providecommand{\bibinfo}[2]{#2}
\providecommand{\eprint}[2][]{\url{#2}}

\bibitem[{\citenamefont{Kilian}(1988)}]{kilian:foundingOnOT}
\bibinfo{author}{\bibfnamefont{J.}~\bibnamefont{Kilian}}, in
  \emph{\bibinfo{booktitle}{Proceedings of 20th ACM STOC}}
  (\bibinfo{year}{1988}), pp. \bibinfo{pages}{20--31}.

\bibitem[{\citenamefont{Cr\'epeau et~al.}(1995)\citenamefont{Cr\'epeau, van~de
  Graaf, and Tapp}}]{crepeau:committedOT}
\bibinfo{author}{\bibfnamefont{C.}~\bibnamefont{Cr\'epeau}},
  \bibinfo{author}{\bibfnamefont{J.}~\bibnamefont{van~de Graaf}},
  \bibnamefont{and} \bibinfo{author}{\bibfnamefont{A.}~\bibnamefont{Tapp}}, in
  \emph{\bibinfo{booktitle}{CRYPTO '95: Proceedings of the 15th Annual
  International Cryptology Conference on Advances in Cryptology}}
  (\bibinfo{publisher}{Springer-Verlag}, \bibinfo{year}{1995}), pp.
  \bibinfo{pages}{110--123}.

\bibitem[{\citenamefont{Wiesner}(1983)}]{wiesner:conjugate}
\bibinfo{author}{\bibfnamefont{S.}~\bibnamefont{Wiesner}},
  \bibinfo{journal}{Sigact News} \textbf{\bibinfo{volume}{15}}
  (\bibinfo{year}{1983}).

\bibitem[{\citenamefont{Rabin}(1981)}]{rabin:ot}
\bibinfo{author}{\bibfnamefont{M.}~\bibnamefont{Rabin}}, \bibinfo{type}{Tech.
  Rep.}, \bibinfo{institution}{Aiken Computer Laboratory, Harvard University}
  (\bibinfo{year}{1981}), \bibinfo{note}{technical Report TR-81}.

\bibitem[{\citenamefont{Even et~al.}(1985)\citenamefont{Even, Goldreich, and
  Lempel}}]{even:firstOT}
\bibinfo{author}{\bibfnamefont{S.}~\bibnamefont{Even}},
  \bibinfo{author}{\bibfnamefont{O.}~\bibnamefont{Goldreich}},
  \bibnamefont{and} \bibinfo{author}{\bibfnamefont{A.}~\bibnamefont{Lempel}},
  \bibinfo{journal}{Communications of the ACM} \textbf{\bibinfo{volume}{28}},
  \bibinfo{pages}{637} (\bibinfo{year}{1985}).

\bibitem[{\citenamefont{Cr{\'e}peau}(1994)}]{crepeau:qot}
\bibinfo{author}{\bibfnamefont{C.}~\bibnamefont{Cr{\'e}peau}},
  \bibinfo{journal}{Journal of Modern Optics} \textbf{\bibinfo{volume}{41}},
  \bibinfo{pages}{2455} (\bibinfo{year}{1994}).

\bibitem[{\citenamefont{Bennett et~al.}(1992)\citenamefont{Bennett, Brassard,
  Cr\'epeau, and Skubiszewska}}]{crepeau:practicalOT}
\bibinfo{author}{\bibfnamefont{C.~H.} \bibnamefont{Bennett}},
  \bibinfo{author}{\bibfnamefont{G.}~\bibnamefont{Brassard}},
  \bibinfo{author}{\bibfnamefont{C.}~\bibnamefont{Cr\'epeau}},
  \bibnamefont{and} \bibinfo{author}{\bibfnamefont{M.-H.}
  \bibnamefont{Skubiszewska}}, in \emph{\bibinfo{booktitle}{CRYPTO '91:
  Proceedings of the 11th Annual International Cryptology Conference on
  Advances in Cryptology}} (\bibinfo{publisher}{Springer-Verlag},
  \bibinfo{year}{1992}), pp. \bibinfo{pages}{351--366}.

\bibitem[{\citenamefont{Lo}(1997)}]{lo:insecurity}
\bibinfo{author}{\bibfnamefont{H.-K.} \bibnamefont{Lo}},
  \bibinfo{journal}{Physical Review A} \textbf{\bibinfo{volume}{56}},
  \bibinfo{pages}{1154} (\bibinfo{year}{1997}),
  \bibinfo{note}{quant-ph/9611031}.

\bibitem[{\citenamefont{Mayers}(1996)}]{mayers:trouble}
\bibinfo{author}{\bibfnamefont{D.}~\bibnamefont{Mayers}}
  (\bibinfo{year}{1996}), \bibinfo{note}{quant-ph/9603015}.

\bibitem[{\citenamefont{Lo and Chau}(1997)}]{lo&chau:bitcom}
\bibinfo{author}{\bibfnamefont{H.-K.} \bibnamefont{Lo}} \bibnamefont{and}
  \bibinfo{author}{\bibfnamefont{H.~F.} \bibnamefont{Chau}},
  \bibinfo{journal}{Physical Review Letters} \textbf{\bibinfo{volume}{78}},
  \bibinfo{pages}{3410} (\bibinfo{year}{1997}),
  \bibinfo{note}{quant-ph/9603004}.

\bibitem[{\citenamefont{Mayers}(1997)}]{mayers:bitcom}
\bibinfo{author}{\bibfnamefont{D.}~\bibnamefont{Mayers}},
  \bibinfo{journal}{Physical Review Letters} \textbf{\bibinfo{volume}{78}},
  \bibinfo{pages}{3414} (\bibinfo{year}{1997}),
  \bibinfo{note}{quant-ph/9605044}.

\bibitem[{\citenamefont{Lo and Chau}(1996)}]{lo&chau:bitcom2}
\bibinfo{author}{\bibfnamefont{H.-K.} \bibnamefont{Lo}} \bibnamefont{and}
  \bibinfo{author}{\bibfnamefont{H.}~\bibnamefont{Chau}}, in
  \emph{\bibinfo{booktitle}{Proceedings of PhysComp96}} (\bibinfo{year}{1996}),
  \bibinfo{note}{quant-ph/9605026}.

\bibitem[{\citenamefont{Damg{\aa}rd et~al.}(2005)\citenamefont{Damgaard, Fehr,
  Salvail, and Schaffner}}]{serge:bounded}
\bibinfo{author}{\bibfnamefont{I.}~\bibnamefont{Damgaard}},
  \bibinfo{author}{\bibfnamefont{S.}~\bibnamefont{Fehr}},
  \bibinfo{author}{\bibfnamefont{L.}~\bibnamefont{Salvail}}, \bibnamefont{and}
  \bibinfo{author}{\bibfnamefont{C.}~\bibnamefont{Schaffner}}, in
  \emph{\bibinfo{booktitle}{Proceedings of 46th IEEE FOCS}}
  (\bibinfo{year}{2005}), pp. \bibinfo{pages}{449--458}.

\bibitem[{\citenamefont{Damg{\aa}rd et~al.}(2007)\citenamefont{Damg{\aa}rd,
  Fehr, Renner, Salvail, and Schaffner}}]{serge:new}
\bibinfo{author}{\bibfnamefont{I.~B.} \bibnamefont{Damg{\aa}rd}},
  \bibinfo{author}{\bibfnamefont{S.}~\bibnamefont{Fehr}},
  \bibinfo{author}{\bibfnamefont{R.}~\bibnamefont{Renner}},
  \bibinfo{author}{\bibfnamefont{L.}~\bibnamefont{Salvail}}, \bibnamefont{and}
  \bibinfo{author}{\bibfnamefont{C.}~\bibnamefont{Schaffner}}, in
  \emph{\bibinfo{booktitle}{Advances in Cryptology---CRYPTO~'07}}
  (\bibinfo{publisher}{Springer-Verlag}, \bibinfo{year}{2007}), vol.
  \bibinfo{volume}{4622} of \emph{\bibinfo{series}{Lecture Notes in Computer
  Science}}, pp. \bibinfo{pages}{360--378}, \eprint{quant-ph/0612014}.

\bibitem[{\citenamefont{Bennett and Brassard}(1984)}]{BB:84}
\bibinfo{author}{\bibfnamefont{C.~H.} \bibnamefont{Bennett}} \bibnamefont{and}
  \bibinfo{author}{\bibfnamefont{G.}~\bibnamefont{Brassard}}, in
  \emph{\bibinfo{booktitle}{Proceedings of the IEEE International Conference on
  Computers, Systems and Signal Processing}} (\bibinfo{year}{1984}), pp.
  \bibinfo{pages}{175--179}.

\bibitem[{\citenamefont{Gisin et~al.}(2002)\citenamefont{Gisin, Ribordy,
  Tittel, and Zbinden}}]{GRTZ:qkd_review}
\bibinfo{author}{\bibfnamefont{N.}~\bibnamefont{Gisin}},
  \bibinfo{author}{\bibfnamefont{G.}~\bibnamefont{Ribordy}},
  \bibinfo{author}{\bibfnamefont{W.}~\bibnamefont{Tittel}}, \bibnamefont{and}
  \bibinfo{author}{\bibfnamefont{H.}~\bibnamefont{Zbinden}},
  \bibinfo{journal}{Reviews of Modern Physics} \textbf{\bibinfo{volume}{74}},
  \bibinfo{pages}{pp. 145} (\bibinfo{year}{2002}).

\bibitem[{\citenamefont{Damgaard et~al.}(2007)\citenamefont{Damgaard, Fehr,
  Salvail, and Schaffner}}]{DFSS:secureid}
\bibinfo{author}{\bibfnamefont{I.}~\bibnamefont{Damgaard}},
  \bibinfo{author}{\bibfnamefont{S.}~\bibnamefont{Fehr}},
  \bibinfo{author}{\bibfnamefont{L.}~\bibnamefont{Salvail}}, \bibnamefont{and}
  \bibinfo{author}{\bibfnamefont{C.}~\bibnamefont{Schaffner}},
  \bibinfo{journal}{LNCS} \textbf{\bibinfo{volume}{4622}}, \bibinfo{pages}{342}
  (\bibinfo{year}{2007}), \bibinfo{note}{arxiv:0708.2557}.

\bibitem[{\citenamefont{Schaffner}(2007)}]{chris:diss}
\bibinfo{author}{\bibfnamefont{C.}~\bibnamefont{Schaffner}}, Ph.D. thesis,
  \bibinfo{school}{University of Aarhus} (\bibinfo{year}{2007}),
  \bibinfo{note}{http://arxiv.org/abs/0709.0289}.

\bibitem[{\citenamefont{Damg{\aa}rd et~al.}(2008)\citenamefont{Damg{\aa}rd,
  Fehr, Salvail, and Schaffner}}]{DFSS08journal}
\bibinfo{author}{\bibfnamefont{I.~B.} \bibnamefont{Damg{\aa}rd}},
  \bibinfo{author}{\bibfnamefont{S.}~\bibnamefont{Fehr}},
  \bibinfo{author}{\bibfnamefont{L.}~\bibnamefont{Salvail}}, \bibnamefont{and}
  \bibinfo{author}{\bibfnamefont{C.}~\bibnamefont{Schaffner}},
  \bibinfo{journal}{special issue of SIAM Journal of Computing}
  (\bibinfo{year}{2008}), \bibinfo{note}{to appear}.

\bibitem[{\citenamefont{Cr\'epeau and Kilian}(1988)}]{crepeau:weakenedOT}
\bibinfo{author}{\bibfnamefont{C.}~\bibnamefont{Cr\'epeau}} \bibnamefont{and}
  \bibinfo{author}{\bibfnamefont{J.}~\bibnamefont{Kilian}}, in
  \emph{\bibinfo{booktitle}{Proceedings of 29th IEEE FOCS}}
  (\bibinfo{year}{1988}).

\bibitem[{\citenamefont{Cr{\'e}peau et~al.}(2004)\citenamefont{Cr{\'e}peau,
  Morozov, and Wolf}}]{CMW04}
\bibinfo{author}{\bibfnamefont{C.}~\bibnamefont{Cr{\'e}peau}},
  \bibinfo{author}{\bibfnamefont{K.}~\bibnamefont{Morozov}}, \bibnamefont{and}
  \bibinfo{author}{\bibfnamefont{S.}~\bibnamefont{Wolf}}, in
  \emph{\bibinfo{booktitle}{International Conference on Security in
  Communication Networks (SCN)}} (\bibinfo{year}{2004}),
  vol.~\bibinfo{volume}{4} of \emph{\bibinfo{series}{Lecture Notes in Computer
  Science}}.

\bibitem[{\citenamefont{Salvail}(1998)}]{salvail:physical}
\bibinfo{author}{\bibfnamefont{L.}~\bibnamefont{Salvail}}, in
  \emph{\bibinfo{booktitle}{Proceedings of CRYPTO'98}} (\bibinfo{year}{1998}),
  vol. \bibinfo{volume}{1462} of \emph{\bibinfo{series}{Lecture Notes in
  Computer Science}}, pp. \bibinfo{pages}{338--353}.

\bibitem[{\citenamefont{Carter and Wegman}(1979)}]{CarWeg79}
\bibinfo{author}{\bibfnamefont{J.~L.} \bibnamefont{Carter}} \bibnamefont{and}
  \bibinfo{author}{\bibfnamefont{M.~N.} \bibnamefont{Wegman}},
  \bibinfo{journal}{Journal of Computer and System Sciences}
  \textbf{\bibinfo{volume}{18}}, \bibinfo{pages}{143} (\bibinfo{year}{1979}).

\bibitem[{\citenamefont{Renner}(2005)}]{renato:diss}
\bibinfo{author}{\bibfnamefont{R.}~\bibnamefont{Renner}}, Ph.D. thesis,
  \bibinfo{school}{ETH Zurich} (\bibinfo{year}{2005}),
  \bibinfo{note}{quant-ph/0512258}.

\bibitem[{\citenamefont{Renner and K\"onig}(2005)}]{renato:compose}
\bibinfo{author}{\bibfnamefont{R.}~\bibnamefont{Renner}} \bibnamefont{and}
  \bibinfo{author}{\bibfnamefont{R.}~\bibnamefont{K\"onig}}, in
  \emph{\bibinfo{booktitle}{Proceedings of TCC 2005}}
  (\bibinfo{publisher}{Springer}, \bibinfo{year}{2005}), vol.
  \bibinfo{volume}{3378} of \emph{\bibinfo{series}{Lecture Notes in Computer
  Science}}, pp. \bibinfo{pages}{407--425}.

\bibitem[{\citenamefont{Buhrman et~al.}(2006)\citenamefont{Buhrman, Christandl,
  Hayden, Lo, and Wehner}}]{wehner06d}
\bibinfo{author}{\bibfnamefont{H.}~\bibnamefont{Buhrman}},
  \bibinfo{author}{\bibfnamefont{M.}~\bibnamefont{Christandl}},
  \bibinfo{author}{\bibfnamefont{P.}~\bibnamefont{Hayden}},
  \bibinfo{author}{\bibfnamefont{H.-K.} \bibnamefont{Lo}}, \bibnamefont{and}
  \bibinfo{author}{\bibfnamefont{S.}~\bibnamefont{Wehner}},
  \bibinfo{journal}{Physical Review Letters} \textbf{\bibinfo{volume}{97}},
  \bibinfo{pages}{250501} (\bibinfo{year}{2006}), \eprint{quant-ph/0609237}.

\bibitem[{\citenamefont{Cr\'epeau}(1997)}]{crepeau:efficientOT}
\bibinfo{author}{\bibfnamefont{C.}~\bibnamefont{Cr\'epeau}}, in
  \emph{\bibinfo{booktitle}{Advances in Cryptology -- Proceedings of EUROCRYPT
  '97}} (\bibinfo{year}{1997}).

\bibitem[{\citenamefont{Julsgaard et~al.}(2004)\citenamefont{Julsgaard,
  Sherson, Cirac, Fiurasek, and Polzik}}]{julsgaard+:mem}
\bibinfo{author}{\bibfnamefont{B.}~\bibnamefont{Julsgaard}},
  \bibinfo{author}{\bibfnamefont{J.}~\bibnamefont{Sherson}},
  \bibinfo{author}{\bibfnamefont{J.~I.} \bibnamefont{Cirac}},
  \bibinfo{author}{\bibfnamefont{J.}~\bibnamefont{Fiurasek}}, \bibnamefont{and}
  \bibinfo{author}{\bibfnamefont{E.~S.} \bibnamefont{Polzik}},
  \bibinfo{journal}{Nature} \textbf{\bibinfo{volume}{432}}, \bibinfo{pages}{pp.
  482} (\bibinfo{year}{2004}).

\bibitem[{\citenamefont{Boozer et~al.}(2007)\citenamefont{Boozer, Boca, Miller,
  Northup, and Kimble}}]{boozer+:qmem}
\bibinfo{author}{\bibfnamefont{A.~D.} \bibnamefont{Boozer}},
  \bibinfo{author}{\bibfnamefont{A.}~\bibnamefont{Boca}},
  \bibinfo{author}{\bibfnamefont{R.}~\bibnamefont{Miller}},
  \bibinfo{author}{\bibfnamefont{T.~E.} \bibnamefont{Northup}},
  \bibnamefont{and} \bibinfo{author}{\bibfnamefont{H.~J.}
  \bibnamefont{Kimble}}, \emph{\bibinfo{title}{Reversible state transfer
  between light and a single trapped atom}} (\bibinfo{year}{2007}),
  \bibinfo{note}{quant-ph/0702248}.

\bibitem[{\citenamefont{Chaneli\`{e}re
  et~al.}(2005)\citenamefont{Chaneli\`{e}re, Matsukevich, Jenkins, Lan,
  Kennedy, and Kuzmich}}]{chaneliere+:qmem}
\bibinfo{author}{\bibfnamefont{T.}~\bibnamefont{Chaneli\`{e}re}},
  \bibinfo{author}{\bibfnamefont{D.}~\bibnamefont{Matsukevich}},
  \bibinfo{author}{\bibfnamefont{S.}~\bibnamefont{Jenkins}},
  \bibinfo{author}{\bibfnamefont{S.-Y.} \bibnamefont{Lan}},
  \bibinfo{author}{\bibfnamefont{T.}~\bibnamefont{Kennedy}}, \bibnamefont{and}
  \bibinfo{author}{\bibfnamefont{A.}~\bibnamefont{Kuzmich}},
  \bibinfo{journal}{Nature} \textbf{\bibinfo{volume}{438}}, \bibinfo{pages}{pp.
  833} (\bibinfo{year}{2005}).

\bibitem[{\citenamefont{Eisaman et~al.}(2005)\citenamefont{Eisaman, Andr\'{e},
  Massou, Fleischauer, Zibrov, and Lukin}}]{eisaman+:qmem}
\bibinfo{author}{\bibfnamefont{M.}~\bibnamefont{Eisaman}},
  \bibinfo{author}{\bibfnamefont{A.}~\bibnamefont{Andr\'{e}}},
  \bibinfo{author}{\bibfnamefont{F.}~\bibnamefont{Massou}},
  \bibinfo{author}{\bibfnamefont{M.}~\bibnamefont{Fleischauer}},
  \bibinfo{author}{\bibfnamefont{A.}~\bibnamefont{Zibrov}}, \bibnamefont{and}
  \bibinfo{author}{\bibfnamefont{M.~D.} \bibnamefont{Lukin}},
  \bibinfo{journal}{Nature} \textbf{\bibinfo{volume}{438}}, \bibinfo{pages}{pp.
  837} (\bibinfo{year}{2005}).

\bibitem[{\citenamefont{Rosenfeld et~al.}(2007)\citenamefont{Rosenfeld, Berner,
  Volz, Weber, and Weinfurter}}]{rosenfeld+:qmem}
\bibinfo{author}{\bibfnamefont{W.}~\bibnamefont{Rosenfeld}},
  \bibinfo{author}{\bibfnamefont{S.}~\bibnamefont{Berner}},
  \bibinfo{author}{\bibfnamefont{J.}~\bibnamefont{Volz}},
  \bibinfo{author}{\bibfnamefont{M.}~\bibnamefont{Weber}}, \bibnamefont{and}
  \bibinfo{author}{\bibfnamefont{H.}~\bibnamefont{Weinfurter}},
  \bibinfo{journal}{Physical Review Letters} \textbf{\bibinfo{volume}{98}},
  \bibinfo{pages}{0505004} (\bibinfo{year}{2007}).

\bibitem[{\citenamefont{Pittman and Franson}(2002)}]{pittman_franson:memory}
\bibinfo{author}{\bibfnamefont{T.~B.} \bibnamefont{Pittman}} \bibnamefont{and}
  \bibinfo{author}{\bibfnamefont{J.~D.} \bibnamefont{Franson}},
  \bibinfo{journal}{Phys. Rev. A} \textbf{\bibinfo{volume}{66}},
  \bibinfo{pages}{062302} (\bibinfo{year}{2002}).

\bibitem[{\citenamefont{Knill et~al.}(2001)\citenamefont{Knill, Laflamme, and
  Milburn}}]{KLM:lo}
\bibinfo{author}{\bibfnamefont{E.}~\bibnamefont{Knill}},
  \bibinfo{author}{\bibfnamefont{R.}~\bibnamefont{Laflamme}}, \bibnamefont{and}
  \bibinfo{author}{\bibfnamefont{G.}~\bibnamefont{Milburn}},
  \bibinfo{journal}{Nature} \textbf{\bibinfo{volume}{409}}, \bibinfo{pages}{46}
  (\bibinfo{year}{2001}), \eprint{\url{http://arxiv.org/abs/quant-ph/0006088}}.

\bibitem[{\citenamefont{Wehner and Wullschleger}(2007)}]{WW07:compose}
\bibinfo{author}{\bibfnamefont{S.}~\bibnamefont{Wehner}} \bibnamefont{and}
  \bibinfo{author}{\bibfnamefont{J.}~\bibnamefont{Wullschleger}}
  (\bibinfo{year}{2007}), \bibinfo{note}{arxiv:0709.0492}.

\bibitem[{\citenamefont{Vandenberghe and Boyd}(1996)}]{VB:sp}
\bibinfo{author}{\bibfnamefont{L.}~\bibnamefont{Vandenberghe}}
  \bibnamefont{and} \bibinfo{author}{\bibfnamefont{S.}~\bibnamefont{Boyd}},
  \bibinfo{journal}{SIAM review} \textbf{\bibinfo{volume}{38}},
  \bibinfo{pages}{49} (\bibinfo{year}{1996}).

\bibitem[{\citenamefont{Helstrom}(1967)}]{helstrom:detection}
\bibinfo{author}{\bibfnamefont{C.~W.} \bibnamefont{Helstrom}},
  \bibinfo{journal}{Information and Control} \textbf{\bibinfo{volume}{10}},
  \bibinfo{pages}{254} (\bibinfo{year}{1967}).

\bibitem[{\citenamefont{Hayashi}(2006)}]{hayashi}
\bibinfo{author}{\bibfnamefont{M.}~\bibnamefont{Hayashi}},
  \emph{\bibinfo{title}{Quantum Information - An introduction}}
  (\bibinfo{publisher}{Springer}, \bibinfo{year}{2006}).

\bibitem[{\citenamefont{Horn and Johnson}(1985)}]{horn&johnson:ma}
\bibinfo{author}{\bibfnamefont{R.~A.} \bibnamefont{Horn}} \bibnamefont{and}
  \bibinfo{author}{\bibfnamefont{C.~R.} \bibnamefont{Johnson}},
  \emph{\bibinfo{title}{Matrix Analysis}} (\bibinfo{publisher}{Cambridge
  University Press}, \bibinfo{year}{1985}).

\end{thebibliography}

\appendix

\section{Tools}
In this appendix, we prove the lemmas used in the main text.
The statements are reproduced for convenience.\\[2mm]
\noindent {\bf Lemma~\ref{lemma:individual} }{\it
Let $\rho_{XE}$ be a cq-state with uniformly distributed $X \in \01^n$ and $\rho_E^x =
\rho_{E_1}^{x_1} \otimes \ldots \otimes \rho_{E_n}^{x_n}$.
Then the maximum probability of guessing $x$ given state $\rho_E$ is
$P_g(X|\rho_E)  = \Pi_{i=1}^n P_g(X_i|\rho_{E_i})$,
which can be achieved by measuring each register separately.}\\[2mm]
\begin{proof}
  For simplicity, we will assume that each bit is encoded using the
  same states $\rho_0 = \rho_{E_i}^0$ and $\rho_1 = \rho_{E_i}^1$.  The argument
  for different encodings is analogous, but harder to read.  First of
  all, note that we can phrase the problem of finding the optimal
  probability of distinguishing two states as a semi-definite program
  (SDP)
\begin{sdp}{maximize}{$\frac{1}{2}\left(\Tr(M_0\rho_0) + \Tr(M_1\rho_1)\right)$}
&$M_0,M_1 \geq 0$\\
&$M_0 + M_1 = \id$
\end{sdp}
with the dual program
\begin{sdp}{minimize}{$\frac{1}{2}\Tr(Q)$}
&$Q \geq \rho_0$\\
&$Q \geq \rho_1$.
\end{sdp}
Let $p_*$ and $d_*$ denote the optimal values of the primal and dual
respectively. From the weak duality of SDPs, we have $p_* \leq d_*$.
Indeed, since $M_0,M_1 = \id/2$ are feasible solutions, we even have
strong duality: $p_* = d_*$ \cite{VB:sp}.

Of course, the problem of determining the entire string $x$ from
$\hat{\rho}_x \assign \rho_E^x$ can
also be phrased as a SDP:
\begin{sdp}{maximize}{$\frac{1}{2^n} \sum_{x \in \01^n} \Tr(M_x\hat{\rho}_x)$}
&$\forall x, M_x \geq 0$\\
&$\sum_{x \in \01^n} M_x = \id$
\end{sdp}
with the corresponding dual
\begin{sdp}{minimize}{$\frac{1}{2^n} \Tr(\hat{Q})$}
&$\forall x, \hat{Q} \geq \hat{\rho}_x$.
\end{sdp}
Let $\hat{p}_*$ and $\hat{d}_*$ denote the optimal values of this new primal and dual
respectively. Again, $\hat{p}_* = \hat{d}_*$.

Note that when trying to learn the entire string $x$, we are of course free
to measure each register individually and thus $(p_*)^n \leq \hat{p}_*$. We now
show that $\hat{d}_* \leq (d_*)^n$ by constructing a dual solution $\hat{Q}$ from the optimal
solution to the dual of the single-register case, $Q_*$:
Take $\hat{Q} = Q_*^{\otimes n}$. Since $Q_* \geq \rho_0$ and $Q_* \geq \rho_1$
it follows that $\forall x, Q_*^{\otimes n} \geq \hat{\rho}_x$. Thus $\hat{Q}$ is satisfies
the dual constraints. Clearly, $2^{-n} \Tr(\hat{Q}) = (2^{-1} \Tr(Q_*))^n$ and thus we
have $\hat{d}_* \leq (d_*)^n$ as promised. But from $(p_*)^n \leq \hat{p}_*$,
$\hat{p}_* = \hat{d}_*$, and $p_* = d_*$ we immediately have $\hat{p}_* = (p_*)^n$.
\end{proof}

\noindent {\bf Lemma~\ref{lem:iso} }{\it
The only superoperators ${\cal S}\colon \mathbb{C}_2 \rightarrow \mathbb{C}_k$ for which
\begin{equation}
P_g(X|{\cal S}(\sigma_+)) \cdot P_g(X|{\cal S}(\sigma_{\times}))=1,
\end{equation}
are reversible.} \\[2mm]
\begin{proof}
Using Helstrom's formula \cite{helstrom:detection} we have that
$P_g(Z|\mS(\sigma_b))=\frac{1}{2}[1+||{\cal S}(\sigma_{0,b})-{\cal S}(\sigma_{1,b})||_{\tr}/2]$ and thus for $\Delta({\cal S})=1$
we need that for both $b \in \{\times, +\}$, $||{\cal S}(\sigma_{0,b})-{\cal S}(\sigma_{1,b})||_{\tr}/2=1$. This implies that
${\cal S}(\sigma_{0,b})$ and ${\cal S}(\sigma_{1,b})$ are states which have support on orthogonal
sub-spaces for {\em both} $b$. Let ${\cal S}(\sigma_{0,+})=\sum_k p_k \ket{\psi_k}\bra{\psi_k}$ and
${\cal S}(\sigma_{1,+})=\sum_k q_k \ket{\psi_k^{\perp}}\bra{\psi_k^{\perp}}$ where for all $k,l$ $\bra{\psi_k^{\perp}}\psi_l\rangle=0$.
Consider the purification of ${\cal S}(\sigma_{i,b})$ using an ancillary system
i.e. $\ket{\phi_{i,b}}=U_{{\cal S}} \ket{i}_b\ket{0}$. We can write
$\ket{\phi_{0,+}}=\sum_k \sqrt{p_k} \ket{\psi_k,k}$ and $\ket{\phi_{1,+}}=\sum_k \sqrt{q_k} \ket{\psi_k^{\perp},k}$.
Hence $U_{{\cal S}}\ket{0}_{\times}\ket{0}=\frac{1}{\sqrt{2}}(\ket{\phi_{0,+}}+\ket{\phi_{1,+}})$ and similar for
$U_{\cal S}\ket{1}_{\times} \ket{0}$. So we can write
\begin{eqnarray*}
||{\cal S}(\sigma_{0,\times})-{\cal S}(\sigma_{1,\times})||_{\tr}& =& \\
||\sum_k \sqrt{p_k q_k} (\ket{\psi_k}\bra{\psi_k^{\perp}}+\ket{\psi_k^{\perp}}\bra{\psi_k})||_{\tr} & \leq & \\
2 \sum_k \sqrt{p_k q_k}.
\end{eqnarray*}
For this quantity to be equal to 2 we observe that it is necessary that $p_k=q_k$. Thus we set $p_k=q_k$. Then
we observe that if any of the states $\ket{\psi_k}$ (or $\psi_k^{\perp}$) are non-orthogonal, i.e. $|\bra{\psi_k} \psi_l \rangle|>0$, then
the quantity $||\sum_k p_k (\ket{\psi_k}\bra{\psi_k^{\perp}}+\ket{\psi_k^{\perp}}\bra{\psi_k})||_{\tr} < 2$.

Let $S_k$ be the two-dimensional subspace spanned by the orthogonal vectors $\ket{\psi_k}$ and $\ket{\psi_k^{\perp}}$.
By the arguments above, the spaces $S_k$ are mutually orthogonal. We can reverse the super-operator ${\cal S}$ by first
projecting the output into one of the orthogonal subspaces $S_k$ and then applying a unitary operator $U_k$
that maps $\ket{\psi_k}$ and $\ket{\psi_k^{\perp}}$ onto the states $\ket{0}$ and $\ket{1}$.
\end{proof}

\begin{lemma} \label{lem:bestdistribution}
For any $\frac12 \leq p_i \leq 1$ with $\prod_{i=1}^n p_i \leq p^n$, we
have
\begin{equation} \label{eq:product}
\frac{1}{2^{n}} \prod_{i=1}^n (1+p_i) \leq p^{\log(4/3) n} \, .
\end{equation}
\end{lemma}
\begin{proof}
With $\lambda \assign \log(4/3)$, it is easy to verify that $ p_i^{-\lambda} +
  p_i^{1-\lambda} \leq 2$ for $1/2 \leq p_i \leq 1$ and therefore,
\begin{eqnarray*}
\frac{1}{2^{n}}\prod_{i=1}^n (1+p_i) &=& \frac{1}{2^{n}} \prod_{i=1}^n p_i^\lambda \left( p_i^{-\lambda} +
  p_i^{1-\lambda} \right)\\
&\leq& \frac{1}{2^{n}} \cdot p^{\lambda n} \cdot 2^n.
\end{eqnarray*}
\end{proof}

\section{Depolarizing noise}
\label{app:depol}

We now evaluate $\max_{\mS} \Delta(\mS)^2$ for depolarizing noise. Recall that to determine this quantity,
we have to find an uncertainty relation, Eq.~(\ref{eq:uncert}), by
optimizing over all possible partial measurements $\mP$ as depicted in Figure~\ref{figure:depolModel}.
$$
\Delta^2 \assign \max_{\mS} \Delta(\mS)^2 =
\max_{\mP} P_g(X|{\cal S}(\sigma_+)) \cdot P_g(X|{\cal S}(\sigma_\times)),
$$
where $\mS$ acts on a single qubit, but we drop the index $i$ to improve readability.
For our analysis, it is convenient to think of $\mP$ as a partial measurement of the incoming qubit.
Note that this corresponds to letting Bob perform an arbitrary CPTP map from the space of the incoming
qubit to the space carrying the stored qubit. Furthermore, it is convenient
to consider maximizing the sum instead of the product of guessing probabilities
$$
\Gamma = \max_{\mP} P_g(X|{\cal S}(\sigma_+)) + P_g(X|{\cal S}(\sigma_\times)).
$$
This immediately gives us the bound $\Delta \leq \Gamma/2$. In the following, we will use the shorthand
\begin{eqnarray*}
p_+ &=& P_g(X|\mS(\sigma_+)),\\
p_\times &=& P_g(X|\mS(\sigma_\times))
\end{eqnarray*}
for the probabilities that Bob correctly decodes the bit after Alice has announced the basis information.

Any intermediate measurement $\mP$ that Bob may perform can be characterized by a set of measurement operators
$\{F_k\}$ such that $\sum_k F_k^\dagger F_k = \id$. Let the post-measurement state when Bob measures $\sigma_{i,b}$,
and obtained outcome $k$, be
$\tilde{\sigma}^k_{i,b}$.

The probability that Bob succeeds in decoding the bit
after the announcement of the basis is given by the average of probabilities (over all outcomes $k$) that
conditioned on the fact that he obtained outcome $k$ he correctly decodes the bit.
That is for $b \in \{+,\times\}$
\begin{eqnarray}
p_b &=&
\sum_k p_{k|b} \left(\frac{1}{2} + \frac{1}{4}||p_{0|kb} N(\tilde{\sigma}^k_{0,b}) - p_{1|kb} N(\tilde{\sigma}^k_{1,b})||_{\tr}\right) \nonumber \\
\label{plainProb}
&=& \frac{1}{2} + \frac{1}{4} \sum_k p_{k|b} ||r (p_{0|kb} \tilde{\sigma}^k_{0,b} - p_{1|kb} \tilde{\sigma}^k_{1,b}) \nonumber \\
&& + (1-r) (p_{0|kb} - p_{1|kb})\id/2||_{\tr},
\end{eqnarray}
where
\begin{eqnarray*}
&&p_{k|b} = \Tr(F_k (\sigma_{0,b} + \sigma_{1,b})F_k^\dagger)/2=\\
&&\Tr(F_k
\frac{\sigma_{0,b}+\sigma_{1,b}}{2} F_k^\dagger)=\frac12 \Tr(F_k
F_k^\dagger)
\end{eqnarray*}
is the probability of obtaining measurement outcome $k$
conditioned on the fact that the basis was $b$ (and we even see from
the above that it is actually independent of $b$),
$\tilde{\sigma}^k_{0,b} = F_k\sigma_{0,b}F_k^\dagger/p_{k|0b}$ is the
post-measurement state for outcome $k$, and $p_{0|kb}$ is the
probability that we are given this state. Definitions are analogous for
the bit $1$.

We now show that Bob's optimal strategy is to measure in the Breidbart basis for $r < 1/\sqrt{2}$,
and to simply store the qubit for $r \geq 1/\sqrt{2}$. This then immediately allows us to evaluate $\Delta$.
To prove our result, we proceed in three steps: First, we will simplify our problem considerably until
we are left with a single Hermitian measurement operator over which we need to maximize. Second, we show
that the optimal measurement operator is diagonal in the Breidbart basis. And finally, we show that
depending on the amount of noise, this measurement operator is either proportional to the identity,
or proportional to a rank one projector. Our individual claims are indeed very intuitive.

For any measurement $M = \{F_k\}$, let $B(M) = p_+^M + p_\times^M$ for the measurement $M$, where $p_+^M$ and $p_\times^M$
are the success probabilities similar to Eq.~(\ref{plainProb}), but restricted to using the measurement $M$. First of all, note that we can easily combine two measurements. Intuitively,
the following statement says that if we choose one measurement with probability $\alpha$, and the other
with probability $\beta$ our average success probability will be the average of the success probabilities
obtained via the individual measurements:
\begin{claim}\label{convexity}
Let $M_1=\{F_k^1\}$ and $M_2=\{F_k^2\}$ be two measurements. Then
$B(\alpha M_1 + \beta M_2) = \alpha B(M_1) + \beta B(M_2)$,
where $\alpha M_1 + \beta M_2 := \{\sqrt{\alpha} F_k^1\} \cup \{\sqrt{\beta} F_k^2\}$ for $\alpha,\beta \geq 0$
and $\alpha + \beta=1$.
\end{claim}
\begin{proof}
  Let $F=\set{F_k}_{k=1}^f$ and $G=\set{G_k}_{k=1}^g$ be measurements,
  $0 \leq \alpha \leq 1$ and $M \assign \set{\sqrt{\alpha}
    F_k}_{k=1}^{f} \cup \set{\sqrt{1-\alpha} G_k}_{k=f+1}^{f+g}$ be
  the measurement $F$ with probability $\alpha$ and measurement $G$
  with probability $1-\alpha$. We denote by $p_{\cdot}^F, p_{\cdot}^G,
  p_{\cdot}^M$ the probabilities corresponding to measurements $F,G,M$
  respectively. Observe that for $1\leq k \leq f$, $p_{k|b}^M =
  \frac12 \Tr(\alpha F_k F_k^\dagger) = \alpha p_{k|b}^F$ and
  analogously for $f+1 \leq k \leq f+g$, we have $p_{k|b}^M =
  (1-\alpha)p_{k|b}^G$. We observe furthermore that for $1\leq k \leq
  f$ and $x \in \set{0,1}$, $\alpha$ cancels out by the normalization,
  $\tilde{\sigma}_{x,b}^{k,M} = \frac{\alpha F_k \sigma_{x,b} F_k^\dagger}{p_{k|xb}^M} = \frac{F_k
    \sigma_{x,b} F_k^\dagger}{p_{k|xb}^F} = \tilde{\sigma}_{x,b}^{k,F}$ and similarly for $f+1 \leq k
  \leq f+g$. Finally, we can convince ourselves that
  $p_{x|kb}^M=p_{x|kb}^F=p_{x|(k-f)b}^G$, as the probability to be
  given state $\tilde{\sigma}_{0,b}^k$ is the same when the measurement
  outcome and the basis is fixed.
Putting everything together, we obtain
\begin{align*}
  p_b^M &= \sum_{k=1}^{f+g} p_{k|b}^M \left(\frac{1}{2} +
    \frac{1}{4}||p_{0|kb}^M N(\tilde{\sigma}^{k,M}_{0,b}) - p_{1|kb}^M
    N(\tilde{\sigma}^{k,M}_{1,b})||_{\tr}\right) \\
 &= \sum_{k=1}^{f} \alpha p_{k|b}^F\left(\frac{1}{2} +
    \frac{1}{4}||p_{0|kb}^F N(\tilde{\sigma}^{k,F}_{0,b}) - p_{1|kb}^F
    N(\tilde{\sigma}^{k,F}_{1,b})||_{\tr}\right) \\
 &\quad +\sum_{k=f+1}^g (1-\alpha) p_{k|b}^G  \cdot \\
&\qquad\qquad \left(\frac{1}{2} + \frac{1}{4}||p_{0|kb}^G N(\tilde{\sigma}^{k,G}_{0,b}) - p_{1|kb}^G
    N(\tilde{\sigma}^{k,G}_{1,b})||_{\tr}\right) \\
 &= \alpha p_b^F + (1-\alpha) p_b^G \, .
\end{align*}
\end{proof}

We can now make a series of observations.
\begin{claim}\label{invarianceOfBunderG}
Let $M = \{F_k\}$ and $G = \{\id,X,Z,XZ\}$. Then for all $g \in G$ we have $B(M) = B(gMg^\dagger)$.
\end{claim}
\begin{proof}
This claim follows immediately from that fact that for the trace norm we have $||U A U^\dagger||_{\tr} = ||A||_{\tr}$
for all unitaries $U$, and by noting that for all $g \in G$, $g$ can at most exchange the roles
of $0$ and $1$. That is, we can perform a bit flip before the measurement which we can correct for afterwards
by applying classical post-processing: we have for all $g \in G$ that
\begin{eqnarray*}
&&p_{k|b}||p_{0|kb} N\left(
\frac{F_k g \sigma_{0,b} g^\dagger F_k^\dagger}{p_{k|0b}}\right)
- p_{1|kb} N\left(\frac{F_k g \sigma_{1,b} g^\dagger F_k^\dagger}{p_{k|1b}}\right)||_{\tr}\\
&&=
p_{k'|b}||p_{0|kb} N\left(\frac{F_k \sigma_{0,b} F_k^\dagger}{p_{k|0b}}\right)
- p_{1|kb} N\left(\frac{F_k  \sigma_{1,b}  F_k^\dagger}{p_{k|1b}}\right)||_{\tr}.
\end{eqnarray*}
\end{proof}

It also follows that
\begin{corollary}\label{simplify}
For all $k$ we have for all $b \in \{+,\times\}$ and $g \in G$ that
\begin{eqnarray*}
&&||p_{0|kb} N\left(\frac{F_k \sigma_{0,b} F_k^\dagger}{p_{k|0b}}\right)
- p_{1|kb} N\left(\frac{F_k \sigma_{1,b} F_k^\dagger}{p_{k|1b}}\right)||_{\tr}\\
&& =
||p_{0|kb} N\left(\frac{F_{k} g\sigma_{0,b}g^\dagger F_{k}^\dagger}{p_{k|0b}}\right)
- p_{1|kb} N\left(\frac{F_{k} g \sigma_{1,b}g^\dagger  F_{k}^\dagger}{p_{k|1b}}\right)||_{\tr}.
\end{eqnarray*}
\end{corollary}
\begin{proof}
This follows from the proof of Claim~\ref{invarianceOfBunderG}.
\end{proof}

\begin{claim}\label{maximumInvariant}
Let $G = \{\id,X,Z,XZ\}$.
There exists a measurement operator $F$ such that the maximum of $B(M)$ over all measurements
$M$ is achieved by a measurement proportional to $\{gFg^\dagger \mid g \in G\}$.
\end{claim}
\begin{proof}
Let $M= \{F_k\}$ be a measurement. Let $K = |M|$ be the number of measurement operators.
Clearly, $\hat{M} = \{\hat{F}_{g,k}\}$
with
$$
\hat{F}_{g,k} = \frac{1}{4} g F_k g^\dagger,
$$
is also a quantum measurement since $\sum_{g,k} \hat{F}_{g,k}^\dagger\hat{F}_{g,k} = \id$.
It follows from Claims~\ref{convexity} and~\ref{invarianceOfBunderG} that $B(M) = B(\hat{M})$.
Define operators
$$
N_{g,k} = \frac{1}{\sqrt{2\Tr(F_k^\dagger F_k)}}g F_k g^\dagger.
$$
Note that
$$
\sum_{g \in G} N_{g,k} = \frac{1}{\sqrt{2\Tr(F_k^\dagger F_k)}} \sum_{u,v \in \01}X^uZ^v
F_k^\dagger F_k Z^v X^u = \id.
$$
(see for example Hayashi~\cite{hayashi}).
Hence $M_k = \{N_{g,k}\}$ is a valid quantum measurement.
Now, note that $\hat{M}$ can be obtained from $M_1,\ldots,M_K$ by averaging.
Hence, by Claim~\ref{convexity} we have
$$
B(M)= B(\hat{M}) \leq \max_k B(M_k).
$$
Let $M^*$ be the optimal measurement.
Clearly, $m = B(M^*) \leq \max_k B(M^*_k) \leq m$ by the above and Corollary~\ref{simplify}
from which our claim follows.
\end{proof}

Note that Claim~\ref{maximumInvariant} also gives us that we have at most 4 measurement operators.
Wlog, we will take the measurement outcomes to be labeled $1,2,3,4$.

Finally, we note that we can restrict ourselves to optimizing over positive-semidefinite (and hence Hermitian)
matrices only.
\begin{claim}
Let $F$ be a measurement operator, and let
$g(F) := 1 + \sum_{b,k} p_{k|b}||p_{0|b}N(\tilde{\sigma_{0,b}})- p_{1|b}N(\tilde{\sigma_{1,b}})||_{\tr}$ with
$\tilde{\sigma_{0,b}} = F \sigma_{0,b} F^\dagger/\Tr(F\sigma_{0,b}F^\dagger)$ and
$\tilde{\sigma_{1,b}} = F \sigma_{1,b} F^\dagger/\Tr(F\sigma_{1,b}F^\dagger)$.
Then there exists a Hermitian operator $\hat{F}$, such that $g(F) = g(\hat{F})$.
\end{claim}
\begin{proof}
Let $F^\dagger = \hat{F}U$ be the polar decomposition of $F^\dagger$,
where $\hat{F}$ is positive semidefinite and $U$ is unitary~\cite[Corollary 7.3.3]{horn&johnson:ma}.
Evidently, since the trace is cyclic, all probabilities remain the same. It follows immediately from
the definition of the trace-norm that $||U A U^\dagger||_{\tr} = ||A||_{\tr}$ for all unitaries $U$,
which completes our proof.
\end{proof}

To summarize, our optimization problem can now be simplified to
\begin{eqnarray*}
&&\max_M B(M) = \max_M p_+^M + p_\times^M \leq\\
&&\max_F
1 + \sum_{b,k} p_{k|b} ||p_{0|b}N(\tilde{\sigma_{0,b}})- p_{1|b}N(\tilde{\sigma_{1,b}})||_{\tr}\\
&&= 1 + 2\sum_{b} ||r(F(\sigma_{0,b} - \sigma_{1,b})F) \\
&&\qquad\qquad + (1-r)\Tr(F(\sigma_{0,b} - \sigma_{1,b})F)\frac{\id}{2}||_{\tr}
\end{eqnarray*}
where the maximization is now taken over a single operator $F$, and we have used the fact
that we can write $p_{0|kb} = p_{k|0b}/(2p_{k|b})$ and we have 4 measurement operators.

\subsection{F is diagonal in the Breidbart basis}
Now that we have simplified our problem already considerably, we are ready to perform the actual optimization.
Since we are in $d=2$ and $F$ is Hermitian, we may express $F$ as
$$
F = \alpha \outp{\phi}{\phi} + \beta \outp{\phi^\perp}{\phi^\perp},
$$
for some state $\ket{\phi}$ and real numbers $\alpha,\beta$. We first of all note that from
$\sum_{k} F_k F_k^\dagger = \id$, we obtain that
\begin{eqnarray*}
&&\Tr\left(\sum_k F_k F_k^\dagger\right) = \sum_k \Tr(F_k F_k) =\\
&& \sum_{g \in\{\id,X,Z,XZ\}} \Tr(gFgg^\dagger Fg^\dagger)
= 4 \Tr(FF) = \Tr(\id) = 2,
\end{eqnarray*}
and hence $\Tr(FF) = \alpha^2 + \beta^2 = 1/2$. Furthermore using that
$\outp{\phi}{\phi} + \outp{\phi^\perp}{\phi^\perp} =\id$
we then have
\begin{equation}\label{Fdef}
F = \beta \id + (\alpha - \beta)\outp{\phi}{\phi},
\end{equation}
with $\beta = \sqrt{1-\alpha^2}$. Our first goal is now to show that $\ket{\phi}$ is a Breidbart vector (or the bit-flipped
version thereof). To this end, we first formalize our intuition that we may take $\ket{\phi}$ to lie in the XZ plane of
the Bloch sphere only. Since we are only interested in the trace-distance term of $B(M)$, we restrict ourselves to considering
\begin{eqnarray*}
C(F) \assign \sum_{b}&&||r(F(\sigma_{0,b} - \sigma_{1,b})F) + \\[-2mm]
&&(1-r)\Tr(F(\sigma_{0,b} - \sigma_{1,b})F)\frac{\id}{2}||_{\tr}.
\end{eqnarray*}

\begin{claim}\label{ignoreY}
Let $F$ be the operator that maximizes $C(F)$, and write $F$ as in Eq.(\ref{Fdef}).
Then $\ket{\phi}$ lies in the XZ plane in the Bloch sphere. (i.e. $\Tr(FY)=0$).
\end{claim}
\begin{proof}
We first parametrize the state in terms of its Bloch vector:
$$
\outp{\phi}{\phi} = \frac{\id + x X + y Y + z Z}{2}.
$$
Since $\ket{\phi}$ is pure we can write $y = \sqrt{1 - x^2 - z^2}$. Hence, we can express
$F$ as
$$
F = \frac{1}{2}\left((\alpha + \beta) \id + (\alpha - \beta)(x X + y Y + z Z)\right).
$$
Noting that $\sigma_{0,+} - \sigma_{1,+} = Z$ and $\sigma_{0,\times} - \sigma_{1,\times} = X$
we can compute for the computational basis
\begin{eqnarray*}
P &\assign& r(FZF) + (1-r)\Tr(FZF)\frac{\id}{2}\\ &=&
\frac{1}{2}\left(\left(2\alpha^2 - \frac{1}{2}\right)z\id +
r\left((\alpha-\beta)^2 xz X\right.\right.\\
&+& \left.\left.(\alpha-\beta)^2 yz Y +
\left((\alpha-\beta)^2 z^2 + 2 \alpha \beta\right)
Z\right)\right),
\end{eqnarray*}
and for the Hadamard basis:
\begin{eqnarray*}
T &\assign& r(FXF) + (1-r)\Tr(FXF)\frac{\id}{2}\\
&=&
\frac{1}{2}\left(\left(2\alpha^2-\frac{1}{2}\right)x \id +
r\left(
\left((\alpha-\beta)^2 x^2 + 2 \alpha \beta\right)X\right)\right.\\
&+& \left.
(\alpha-\beta)^2 xy Y + (\alpha-\beta)^2 xz Z\right)
\end{eqnarray*}

Note that $||P||_{\tr} = \sum_j |\lambda_j(P)|$,
where $\lambda_j$ is the $j$-th eigenvalue of $P$.
A lengthy computation (using Mathematica),
and plugging in $\beta = \sqrt{1/2 - \alpha^2}$
and $y = \sqrt{1-x^2-z^2}$ shows that we have
\begin{eqnarray*}
\lambda_1(P) &=& \frac{1}{4}\left(\left(4\alpha^2 - 1\right)z -
r \sqrt{z^2 + 8 \alpha^2(2\alpha^2 - 1)(z^2-1)}\right)\\
\lambda_2(P) &=& \frac{1}{4}\left(\left(4\alpha^2 - 1\right)z +
r \sqrt{z^2 + 8 \alpha^2(2\alpha^2 - 1)(z^2-1)}\right)
\end{eqnarray*}
Similarly, we obtain for the Hadamard basis that
\begin{eqnarray*}
\lambda_1(T) &=& \frac{1}{4}\left(\left(4\alpha^2 - 1\right)x -
r \sqrt{x^2 + 8 \alpha^2(2\alpha^2 - 1)(x^2-1)}\right)\\
\lambda_2(T) &=& \frac{1}{4}\left(\left(4\alpha^2 - 1\right)x +
r \sqrt{x^2 + 8 \alpha^2(2\alpha^2 - 1)(x^2-1)}\right)
\end{eqnarray*}
We define
\begin{eqnarray*}
f(\alpha,x) &\assign& \left(\alpha^2 - \frac{1}{4}\right)x\\
g(\alpha,x) &\assign& \frac{1}{4} \sqrt{x^2 + 8 \alpha^2(2\alpha^2-1)(x^2-1)}.\\
h(\alpha,x,r) &\assign& |f(\alpha,x) + r g(\alpha,x)| + |f(\alpha,x) - r g(\alpha,x)|
\end{eqnarray*}
Note that our optimization problem now takes the form
\begin{sdp}{maximize}{$h(\alpha,x,r) + h(\alpha,z,r)$}
&$x^2+z^2\leq 1$\\
&$0\leq x \leq 1$\\
&$0 \leq z \leq 1$,
\end{sdp}
where we can introduce the last two inequality constraints without loss of generality, since the remaining three
measurement operators will be given by $XFX$, $ZFZ$, and $XZFZX$.

To show that we can let $y=0$ for the optimal solution, we have to show that for all $\alpha$ and all $r$,
the function $h(\alpha,x,r)$ is increasing on the interval $0 \leq x \leq 1$ (and indeed Mathematica
will convince you in an instant that this is the case). Our analysis is further complicated
by the absolute values. We therefore first consider
$$
h(\alpha,x,r)^2 = 2(f(\alpha,x)^2 + r^2 g(\alpha,x)^2 + |f(\alpha,x)^2 - r^2 g(\alpha,x)^2|,
$$
where we have used the fact that $f$ and $g$ are real valued functions. In principle, we can
now analyze
$h_+(\alpha,x,r)^2 = 2(f(\alpha,x)^2 + r^2 g(\alpha,x)^2 + f(\alpha,x)^2 - r^2 g(\alpha,x)^2$
and
$h_-(\alpha,x,r)^2 = 2(f(\alpha,x)^2 + r^2 g(\alpha,x)^2 - f(\alpha,x)^2 + r^2 g(\alpha,x)^2$
separately on their respective domains. By rewriting, we obtain
$$
h_+(\alpha,x,r)^2 = \frac{1}{4}r^2 (x^2 + 8 \alpha^2 (2 \alpha^2 - 1)(x^2-1)),
$$
and
$$
h_-(\alpha,x,r)^2 = 4\left(\alpha^2 - \frac{1}{4}\right)^2 x^2.
$$
Luckily, the first derivatives of
$h_+$ and $h_-$ turns out to be positive everywhere for our choice
of parameters $0 \leq \alpha \leq 1/\sqrt{2}$, and $0 \leq r,z \leq 1$.
Hence, by further inspection at the transitional points
we can conclude that $h$ is an increasing function of $x$.
But this means that to maximize our target expression, we must choose $x$ and
$z$ as large as possible. Hence, choosing $y=0$ is the best choice and our
claim follows.
\end{proof}

We can now immediately extend this analysis to find
\begin{claim}
Let $F$ be the operator that maximizes $C(F)$, and write $F$ as in Eq.~(\ref{Fdef}).
Then
$$
\ket{\phi} = g (\cos(\pi/8)\ket{0} + \sin(\pi/8)\ket{1}),
$$
for some $g \in \{\id,X,Z,XZ\}$.
\end{claim}
\begin{proof}
Extending our analysis from the previous proof, we can compute the second
derivative of both functions. It turns out that also the second
derivatives are positive, and hence $h$ is convex in $x$.
By Claim~\ref{ignoreY}, we can rewrite our optimization problem as
\begin{sdp}{maximize}{$h(\alpha,x,r) + h(\alpha,z,r)$}
&$x^2+z^2 = 1$\\
&$0\leq x \leq 1$\\
&$0 \leq z \leq 1$
\end{sdp}
It now follows from the fact that $h$ is convex in $x$ and the constraint
$x^2 + z^2 = 1$ (by computing the Lagrangian of the above optimization problem),
that for the optimal solution we must have $x=z$, and our claim follows.
\end{proof}

\subsection{Optimality of the trivial strategies}
Now that we have shown that $F$ is in fact diagonal in the Breidbart basis
(or the bit flipped version thereof) we have only a single parameter left in our optimization problem.
We must now optimize over all operators $F$ of the form
$$
F = \alpha \outp{\phi}{\phi} + \sqrt{1/2 - \alpha^2} \outp{\phi^\perp}{\phi^\perp},
$$
where we may take $\ket{\phi}$ to be $\ket{0}_B$ or $\ket{1}_B$. Our aim is now to show that
either $F$ is the identity, or $F = \outp{\phi}{\phi}$ depending on the value of $r$.
\begin{claim}
Let $F$ be the operator that maximizes $C(F)$. Then
$F = c \id$ (for some $c \in \Real$) for $r \geq 1/\sqrt{2}$, and
$F = \outp{\phi}{\phi}$ for $r < 1/\sqrt{2}$, where
$$
\ket{\phi} = g (\cos(\pi/8)\ket{0} + \sin(\pi/8)\ket{1}),
$$
for some $g \in \{\id,X,Z,XZ\}$.
\end{claim}
\begin{proof}
We can now plug in $x = z = 1/\sqrt{2}$ in the expressions for the eigenvalues in our previous proof.
Ignoring the constant factors which do not contribute to our argument, we can then write
\begin{eqnarray*}
\lambda_1(P) &=&
\left(4\alpha^2 - 1\right) -
r \sqrt{1 - 16 \alpha^4 + 8 \alpha^2}\\
\lambda_2(P) &=&
\left(4\alpha^2 - 1\right) +
r \sqrt{1 - 16 \alpha^4 + 8 \alpha^2}\\
\end{eqnarray*}
And similarly for the Hadamard basis.
We again define functions
\begin{eqnarray*}
f(\alpha) &\assign& \left(4\alpha^2 - 1\right)\\
g(\alpha) &\assign& \sqrt{1 -16 \alpha^4 + 8 \alpha^2}\\
h(\alpha,r) &\assign& |f(\alpha,x) + r g(\alpha,x)| + |f(\alpha,x) - r g(\alpha,x)|
\end{eqnarray*}
Note that our optimization problem now takes the form
\begin{sdp}{maximize}{$2h(\alpha,r)$}
&$0\leq \alpha \leq \frac{1}{\sqrt{2}}$
\end{sdp}
Since we are maximizing, we might as well consider the square of our target function
and ignore the leading constant as it is irrelevant for our argument.
$$
h(\alpha,r)^2 = 2(f(\alpha)^2 + r^2 g(\alpha)^2 + |f(\alpha)^2 - r^2 g(\alpha)^2|,
$$
To deal with the absolute value, we now perform a case analysis similar to the one above.
Computing the zeros crossings of the function $f(\alpha)^2 - r^2 g(\alpha)^2$,
we analyze each interval separately. Computing the first and second derivatives on
the intervals we find that $h(\alpha,r)^2$ has exactly two peaks: The first at $\alpha=0$,
and the second at $\alpha=1/2$. We have that $h(0,r)^2 = 2$ for all $r$, and
$h(1/2,r)^2 = 4 r^2$. Hence, we immediately see that the maximum is located at
$\alpha=0$ for $r \leq 1/\sqrt{2}$, and at $\alpha=1/2$ for $r \geq 1/\sqrt{2}$.
\end{proof}

Hence, we may conclude that Bob either measures in the Breidbart
basis, or stores the qubit as is, and Theorem~\ref{thm:depolarize} follows.

We believe that a similar analysis can be done for the dephasing channel, by first symmetrizing
the noise by applying a rotation over $\pi/4$ to our input states.
\end{document}